\documentclass[12pt,letterpaper,journal]{IEEEtran}
\usepackage[cmex10]{amsmath}
\usepackage{amsfonts}
\usepackage{algorithmic,array,cite,mdwmath,mdwtab,eqparbox}
\usepackage{graphicx}
\usepackage{dcolumn}
\usepackage{bm}
\usepackage{amsthm}
\usepackage{subfig}
\newtheorem{example}{Example}[section]
\newtheorem{remark}{Remark}[section]
\newtheorem{theorem}{Theorem}[section]

\newtheorem{lemma}{Lemma}[section]

\newtheorem{definition}{Definition}[section]

\def\ket#1{| #1 \rangle}
\def\bra#1{\langle #1 |}
\def\ip#1#2{\langle #1 | #2 \rangle}

\def\norm#1{\| #1 \|}

\def\diag{\operatorname{diag}}
\def\dim{\operatorname{dim}}
\def\rank{\operatorname{rank}}
\def\Span{\operatorname{span}}

\def\Tr{\operatorname{Tr}}
\def\Ad{\operatorname{Ad}}

\def\B{\mathcal{B}}
\def\C{\mathcal{C}}
\def\F{\mathcal{F}}
\def\H{\mathcal{H}}
\def\M{\mathcal{M}}
\def\O{\mathcal{O}}
\def\T{\mathcal{T}}
\def\RR{\mathbb{R}}
\def\CC{\mathbb{C}}

\def\SU{\mathbb{SU}}

\def\ONE{\mathbb{I}}

\def\su{\mathfrak{su}}

\def\xs{\vec{x}\cdot\vec{\sigma}}

\begin{document}
\title{Analysis of Lyapunov Method for Control of Quantum States}

\author{\IEEEauthorblockN{Xiaoting Wang\IEEEauthorrefmark{1} and
S G Schirmer\IEEEauthorrefmark{1}\\
\IEEEauthorblockA{\IEEEauthorrefmark{1}
                  Dept of Applied Mathematics \& Theoretical Physics,
                  University of Cambridge,\\
                  Wilberforce Road, Cambridge, CB3 0WA, UK\\
           Email: x.wang@damtp.cam.ac.uk, sgs29@cam.ac.uk}}}
\date{\today}
\maketitle

\begin{abstract}
The natural trajectory tracking problem is studied for generic quantum
states represented by density operators.  A control design based on the
Hilbert-Schmidt distance as a Lyapunov function is considered.  The
control dynamics is redefined on an extended space where the LaSalle
invariance principle can be correctly applied even for non-stationary
target states.  LaSalle's invariance principle is used to derive a
general characterization of the invariant set, which is shown to always
contain the critical points of the Lyapunov function.  Critical point
analysis of the latter is used to show that, for generic states, it is a
Morse function with $n!$ isolated critical points, including one global
minimum, one global maximum and $n!-2$ saddles.  It is also shown,
however, that the actual dynamics of the system is not a gradient flow,
and therefore a full eigenvalue analysis of the linearized dynamics
about the critical points of the dynamical system is necessary to
ascertain stability of the critical points.  This analysis shows that a
generic target state is locally asymptotically stable if the linearized
system is controllable and the invariant set is regular, and in fact 
convergence to the target state (trajectory) in this case is almost 
global in that the stable manifolds of all other critical points form a 
subset of measure zero of the state space.  On the other hand, if either 
of these sufficient conditions is not satisfied, the target state ceases 
to be asymptotically stable, a center manifold emerges around the target
state, and the control design ceases to be effective.
\end{abstract}


\section{Introduction}
\label{sec:intro}

Recent advances in quantum optics, trapping of cold atoms, ions and
molecules, and breakthroughs in nano-scale engineering of artificial
atoms and molecules have prompted significant interest in ways to
control these systems and the development of the theoretical foundations
of quantum control theory.  One of the major concerns is how to design
the dynamics of a given system to steer its state to a desired target
and stabilize it in this state.  One proposed technique is Lyapunov
control, where the control function is designed such that a suitably
chosen Lyapunov function is monotonically decreasing along every
trajectory.  A number of Lyapunov-function-based control designs have
been proposed and numerous results established (see, e.g.,
\cite{Vettori,Ferrante,Grivopoulos,Mirrahimi2004a,Mirrahimi2004b,
Mirrahimi2005,Mirrahimi-implicit,altafini1,altafini2,Mirrahimi-stabilize,
D'Alessandro} and references therein).  The convergence analysis in most
of these works is based on the application of LaSalle's invariance
principle.  However, the invariant set for quantum systems is usually
large and thus convergence to the target state cannot be inferred from
the invariance principle directly.

For pure-state quantum systems with states described by wave-functions,
the setting considered in most of the papers to date, there are many
results.  For example, it was shown in~\cite{Mirrahimi2005} that a
particular variant of the method is effective under certain sufficient
conditions, equivalent to controllability of the linearized system.
\cite{Mirrahimi-implicit} also proposed a modified control design based
on an ``implicit'' Lyapunov function for systems that do not satisfy the
conditions for local asymptotic stability.  Although pure states play a
crucial role in quantum mechanics, they form a set of measure zero at
the boundary of the domain of density operators, and in practice quantum
systems are often not in pure states to begin with, due to imperfect
preparation, or because the system is an ensemble of many quantum
particles such as atoms or molecules.  For this reason it is important
to consider control in the context of density operators representing
generic quantum states.  It is also important to consider non-stationary
target states, i.e., target states that are not eigenstates of the
system Hamiltonian, because it is these so-called superposition states
that exhibit non-classical behavior such as interference and
entanglement, which is a crucial resource for novel quantum applications
such as quantum information processing.

This more general case was considered more recently in~\cite{altafini2},
where the convergence properties of the control design were investigated
for both pure and mixed target states, including non-stationary states.
Simulations for bilinear control systems suggest that under certain
strict sufficient conditions as given in~\cite{altafini2}, \emph{all}
trajectories starting outside the invariant set converge to the target
state as opposed to other points in the invariant set.  Simulations also
suggest that when these conditions are not satisfied then we converge to
states (trajectories) in the vicinity of the target, yet never reach the
target, and the states (trajectories) we appear to converge to in the
latter case are \emph{not} critical points of the Lyapunov function.  It
has been suggested that this could be due to numerical errors, but this
fails to explain why such errors should not affect the convergence
equally in all cases.  Can we find analytical results that explain these
observations?  Another issue is the correct application the invariance
principle for non-stationary target states.  Simply transforming to a
rotating frame to make the target state formally appear stationary, as
has been done in several papers, is problematic because the resulting
system is non-autonomous, making the application of the invariance
principle problematic.

To address these issues and evaluate possible explanations e.g., that
the invariant set consists only of critical points of the dynamical
system, which are repulsive except for the target state, necessitates
several steps, including careful characterization of the invariant set,
the set of critical points of the Lyapunov function, and a stability
analysis of the critical points of the dynamical system.  Analysis of
the critical points of the Lyapunov function as considered, e.g., in
\cite{Wu} recently, is not sufficient, as the dynamics is not a gradient
flow.  However, we can still use this analysis to show that the Lyapunov
function is a Morse function if we restrict our attention to generic
states.  This allows us to draw certain conclusions about the dimensions
of the stable (unstable) manifolds of the hyperbolic critical points of
the dynamical system, from which we can draw general conclusions about
the effectiveness of the method in steering the system to the target.

The paper is organized as follows.  In Sec.~\ref{sec:basics} the control
problem is defined and some basic issues such as different notions of
convergence for non-stationary target state are briefly discussed.  In
Sec.~\ref{sec:LaSalle} we formulate the control dynamics as an
autonomous dynamical system defined on an extended space including the
system and the target state, allowing us to apply LaSalle's invariance
principle~\cite{lasalle} to obtain a characterization of the LaSalle
invariant set.  This set is shown to depend on the Hamiltonian as well
as the target state.  In Sec.~\ref{sec:critical} we determine the
critical points of the Lyapunov function on the extended space and
analyze their stability for generic target states.  In
Sec.~\ref{sec:conv_ideal} we analyze the stability of the critical
points of the dynamical system in terms of the eigenvalues of the
linearized system.  Rigorous results are derived for generic stationary
target states and somewhat weaker stability results are given in the
non-stationary case.  In Section~\ref{sec:conv_real} it is shown
explicitly that relaxing the strict requirements on the Hamiltonian
leads to loss of asymptotic stability due to the emergence of center
manifolds about the target state.  Finally, we compare and analyze the
difference between our result and the argument in one preceding
work~\cite{altafini2}.

\section{Natural Trajectory Tracking Problem for Quantum Systems}
\label{sec:basics}

\subsection{Quantum states and evolution}

The state of a quantum system defined on an $n$-dimensional Hilbert
space $\H\simeq \CC^n$ can be represented by a density operator
$\rho$, i.e., an $n\times n$ positive hermitian operator with unit
trace, whose evolution is governed by the following equation:
\begin{eqnarray*}
 i\hbar \dot \rho(t) = [ H, \rho(t) ] =  H\rho(t) - \rho(t) H,
\end{eqnarray*}
where $H$ is an $n\times n$ Hermitian operator representing the system
Hamiltonian, and we shall choose units such that $\hbar=1$.  When the
system is not closed, i.e., interacts with an external environment,
additional terms are required to account for dissipative effects, and 
in the Markovian case the evolution is described by Lindblad master
equation~\cite{Breuer}, for example.  In the following, we will restrict
our analysis to Hamiltonian systems and to an important class of mixed
states, we shall refer to as generic states.  The same analysis can be
applied to density operators representing non-generic states, although
the results will be different.

\begin{definition}
A density operator $\rho$ represents a \emph{generic} state if it
has $n$ distinct eigenvalues.
\end{definition}

\subsection{Control problem}
\label{subsec:problem}

We study the bilinear Hamiltonian control system
\begin{equation}
  \dot \rho(t) = -i [ H_0 + f(t)H_1, \rho(t) ],
\end{equation}
where $f(t)$ is an admissible real-valued control field and $H_0$ and
$H_1$ are a system and control interaction Hamiltonian, respectively,
both of which will be assumed to be time-independent.  The general
control problem is to design a certain control function $f(t)$ such that
the system state $\rho(t)$ with $\rho(0)=\rho_0$ converges to the target
state $\rho_d$ for $t\to\infty$.  We shall assume here that the initial
and target states, $\rho_0$ and $\rho_d$, have the same spectrum because
this is a necessary condition for the target state to be reachable under
unitary evolution.

Since the free Hamiltonian $H_0$ can usually not be turned off, any
target state will evolve according to
\begin{equation}
\label{eqn:3} \dot \rho_d(t) =-i [ H_0, \rho_d(t) ].
\end{equation}
It is easy to see that $\rho_d$ is stationary if and only if it commutes
with $H_0$, $[H_0,\rho_d(0)]=0$.  For any state that does not commute
with $H_0$ the quantum control problem becomes a ``natural tracking
problem''~\cite{Bohacek}, where the objective generally is to find a
control $f(t)$ such that the trajectory $\rho(t)$ with initial state
$\rho_0$ under the controlled evolution asymptotically converges to the
trajectory of $\rho_d(t)$.

\subsection{Trajectory tracking vs orbit tracking}
\label{subsec:tracking}

When the target state $\rho(t)$ is non-stationary, we can have two
different control objectives.  We could require $\rho(t)\to\rho_d(t)$,
as $t\to\infty$, which is known as trajectory tracking; alternatively,
we could require $\rho(t)\to\O(\rho_d(t))$, which is known orbit
tracking, where $\O(\rho_d(t))$ is the orbit of $\rho_d$, defined to be
the set of points $\rho_d(t)$ passes through. By definition, the notion
of orbit tracking is weaker than that of trajectory tracking.  In
particular, the set $\O(\rho_d(t))$ can be very large for states that
follow \emph{non-periodic} trajectories, which is the case for most
states, except for systems of Hilbert space dimension $2$, and higher
dimensional systems for which the eigenvalues of $H_0$ are commensurate
(rational multiples of each other).  Therefore in this work we focus on
quantum state control in the sense of trajectory tracking as this is the
strongest notion of convergence and well-defined for any trajectory.

\section{Lyapunov-based Control}
\label{sec:LaSalle} 

A natural design for the control $f(t)$ is inspired by the construction
of a Lyapunov function $V(\rho,\rho_d)$.  We try to find a control such
that $\frac{d}{dt}V(\rho(t),\rho_d(t))\le 0$, i.e., $V$ decreases along
any trajectory.  If $V(\rho(t),\rho_d(t))$ decreases to zero, we have
$\rho(t)\to \rho_d(t)$.

Define $\M$ to be the set of density operators isospectral with
$\rho_d(0)$. $\M$ is a compact manifold, whose dimension depends on the
spectrum of $\rho_d(0)$.  For a generic $\rho_d$ with $n$ distinct
eigenvalues we have $\M\simeq \mathbb{U}(n)/\mathbb{U}(1)\times \ldots
\times \mathbb{U}(1)$, where the denominator has $n$ factors and $\M$
has dimension $n^2-n$.  Consider the joint dynamics for
$(\rho(t),\rho_d(t))$ on $\M\times\M$:
\begin{subequations}
\label{eqn:auto0}
\begin{align}
\dot \rho(t) &=-i [ H_0+f(t)H_1, \rho(t) ],\\
\dot \rho_d(t) &=-i [ H_0, \rho_d(t) ].
\end{align}
\end{subequations}
The Hilbert-Schmidt norm $\norm{A}=\sqrt{\Tr(A^\dagger A)}$ induces
a natural distance function on $\M\times\M$, which provides a
natural candidate for a Lyapunov function
\begin{equation}
\label{eqn:4}
  V(\rho,\rho_d) = \frac{1}{2}\norm{\rho-\rho_d}^2
                 = \frac{1}{2}\Tr[(\rho-\rho_d)^2].
\end{equation}
Since $\rho$ and $\rho_d$ are required to be isospectral, we have
$\Tr(\rho^2)=\Tr(\rho_d^2)$, and hence
\begin{equation}
 \label{eq:4a}
 V(\rho,\rho_d) = \Tr[\rho_d^2(t)]- \Tr[\rho(t)\rho_d(t)],
\end{equation}
the Lyapunov function used in~\cite{altafini2}.  If
$\rho_d=\ket{\psi_d}\bra{\psi_d}$ and $\rho=\ket{\psi}\bra{\psi}$ we
have furthermore
\begin{equation}
 \label{eq:4b}
 V(\psi,\psi_d) = 1- |\ip{\psi_d(t)}{\psi(t)}|^2,
\end{equation}
a Lyapunov function often used for pure-state
control~\cite{Mirrahimi-implicit}. Choosing
\begin{equation}
 \label{eqn:5}
  f(\rho,\rho_d) = \kappa \Tr([-iH_1,\rho]\rho_d), \quad \kappa>0,
\end{equation}
we find $\dot V(\rho(t),\rho_d(t)) \le 0$.  Without loss of
generality, we set $\kappa=1$ in the following.

Given the target state $\rho_d(0)$, the dynamics under the Lyapunov
control is described by the following nonlinear autonomous dynamical
system on $\M\times\M$:
\begin{subequations}
\label{eqn:auto}
\begin{align}
 \dot{\rho(t)}   &= -i [ H_0+f(\rho,\rho_d)H_1, \rho(t) ],\\
 \dot{\rho_d(t)} &= -i [ H_0, \rho_d(t) ],\\
 f(\rho,\rho_d)  &= \Tr([-iH_1,\rho]\rho_d),
\end{align}
\end{subequations}

\subsection{LaSalle invariance principle}
\label{subsec:invar}

\begin{theorem}
[LaSalle invariance principle~\cite{lasalle}]\label{thm:lasalle:1}
$\dot x=f(x)$ be an autonomous dynamical system, $V(x)$ be a Lyapunov
function on the phase space $\Omega=\{x\}$, satisfying $V(x)>0$ for all
$x \neq x_0$ and $\dot{V}(x) \le 0$, and let $\O(x(t))$ be the orbit of
$x(t)$ in the phase space.  Then the invariant set 
$E=\{\O(x(t))|\dot{V}(x(t))=0\}$ contains the positive limiting sets of
all bounded solutions, i.e., any bounded solution converges to $E$ as
$t\to+\infty$.
\end{theorem}

\begin{remark}
The theorem holds for both real and complex dynamical systems, and what
has basically been proved is that all bounded solutions with $\dot{V}(x)
\ne 0$ converge to the set of solutions with $\dot V(x)=0$.  It does not
matter if $V(x)=0$ for many $x$, as is the case for (\ref{eqn:auto}),
for which $V$ vanishes on the entire set $\{(\rho,\rho_d)\in \M\times\M:
\rho=\rho_d\}$.
\end{remark}

The quantum system~(\ref{eqn:auto}) defined on the extended phase space
$\M\times\M$ is autonomous and any solution $(\rho(t),\rho_d(t))$ is
bounded.  Although the Lyapunov function~(\ref{eq:4a}) is not positive
definite, $V=0$ if and only if $\rho=\rho_d$, which is sufficient to
apply the LaSalle invariance principle~\ref{thm:lasalle:1} to obtain:

\begin{theorem}
\label{thm:lasalle:2} Any trajectory $(\rho(t),\rho_d(t))$ under the
Lyapunov control~(\ref{eqn:5}) will converge to the invariant set
$E=\{(\rho_1,\rho_2) \in \M\times\M: \dot{V}(\rho(t),\rho_d(t))=0,\;
(\rho(0),\rho_d(0))=(\rho_1,\rho_2)\}$.
\end{theorem}

We note here that except when $\rho_d$ is a stationary state, we must
consider the dynamical system on the extended phase space $\M\times\M$
since $V(\rho,\rho_d)$ is \emph{not} well-defined on $\M$. Having
established convergence to the LaSalle invariant set $E$, the next step
is to characterize $E$ for the dynamical system~(\ref{eqn:auto}).

\subsection{Characterization of the LaSalle Invariant Set}
\label{subsec:charE}

LaSalle's invariance principle reduces the convergence analysis to
calculating the invariant set.  It has been argued
(e.g.~\cite{altafini2,D'Alessandro}) that the invariant set consists of
all points $\rho$ that commute with the target state $\rho_d$, i.e.,
$[\rho_d,\rho]=0$.  However, this characterization is only valid for
ideal systems and stationary target states.  Thus we shall first
reconsider the characterization of the invariant set.  Following
standard techniques in nonlinear stability analysis~\cite{Jurdjevic}, we
shall see that the invariant set of the autonomous dynamical system
(\ref{eqn:auto}) defined on the extended state space $\M\times\M$
comprises all pairs $(\rho_1,\rho_2)$ whose commutator is diagonal even
for ideal systems, and it is much larger for non-ideal systems.

It is easy to see that $\dot{V}\equiv 0$ is equivalent to $f(t)\equiv
0$ and a standard argument shows that
\begin{equation*}
\begin{aligned}
0 &= f      = \Tr([-iH_1,\rho]\rho_d)\\
0 &= \dot f = \Tr([-iH_1,\rho]\dot \rho_d)+\Tr([-iH_1,\dot \rho]\rho_d)\\
  & \qquad  = -\Tr([\Ad_{-iH_0}(-iH_1),\rho]\rho_d)\\
  & \cdots \\
0 &= \frac{d^{\ell}f}{dt^\ell}=(-1)^n
\Tr([\Ad^\ell_{-iH_0}(-iH_1),\rho]\rho_d),
\end{aligned}
\end{equation*}
where $B_m=\Ad^m_{-iH_0}(-iH_1)$ represents $m$-fold commutator adjoint
action of $-iH_0$ on $-iH_1$.  Noting $\Tr([A,B]C)=\Tr([A,[B,C]])$ shows
that the LaSalle invariant set consists of all $(\rho_1,\rho_2)$ such that
\begin{equation}
\label{eq:trace-cond1}
 \Tr(B_m[\rho_1,\rho_2])=0.
\end{equation}
If we set $\B^s=\Span\{B_m\}_{m=1}^{m=s}$ and
$\B_0^s=\Span\{B_m\}_{m=0}^{m=s}$ with $B_0=-iH_1$ then
(\ref{eq:trace-cond1}) is equivalent to $[\rho,\rho_d]$ orthogonal to
the subspace $\B=\B_0^\infty$ with respect to the Hilbert-Schmidt inner
product.  The Lie algebra $\su(n)$ can be decomposed into an abelian
part, the Cartan subalgebra $\C=\Span\{\lambda_k\}_{k=1}^{n-1}$,
and an orthogonal subspace $\T$, which is a direct sum of $n(n-1)/2$
root spaces spanned by pairs of generators
$\{\lambda_{k\ell},\bar{\lambda}_{k\ell}\}$ (see
Appendix~\ref{app:lie_algebra_basis}).

Choosing a Hilbert space basis such that $H_0$ is diagonal $H_0=
\diag(a_1,\ldots,a_n)$ with $a_k\ge a_{k+1}$ and
$\omega_{k\ell}=a_\ell-a_k$, which is always possible as $H_0$ is
Hermitian, setting $b_{k\ell}=\alpha_{k\ell}+i\beta_{k\ell}$ and
expanding $-iH_1\in\su(n)$ with respect to the basis~(\ref{eq:lambda})
gives
\begin{equation*}
 \label{eq:H1}
  -i H_1 = \sum_{k=1}^{n-1}\left[ b_k \lambda_k + \sum_{\ell=k+1}^n
           -\alpha_{k\ell} \lambda_{k\ell} + \beta_{k\ell}
           \bar{\lambda}_{k\ell}\right].
\end{equation*}
Noting that $H_0$ is diagonal and using (\ref{eq:lambda_comm}) gives
\begin{equation}
\label{eq:Bm}
\begin{split}
  B_{2m-1} &= \sum_{k=1}^{n-1}\sum_{\ell=k+1}^n \!\!\!(-1)^m
\omega_{k\ell}^{2m-1}
                       [\alpha_{k\ell} \bar{\lambda}_{k\ell}
                        +\beta_{k\ell} \lambda_{k\ell}], \\
  B_{2m}   &= \sum_{k=1}^{n-1}\sum_{\ell=k+1}^n (-1)^m \omega_{k\ell}^{2m}
                       [\alpha_{k\ell} \lambda_{k\ell}
                        -\beta_{k\ell} \bar{\lambda}_{k\ell}].
\end{split}
\end{equation}

\begin{definition}
The Hamiltonian of the dynamical system (\ref{eqn:auto}) is called
ideal if
\begin{itemize}
\item[(i)] $H_0$ is strongly regular, i.e.,
           $\omega_{k\ell}\neq \omega_{pq}$ unless $(k,\ell)=(p,q)$.
\item[(ii)] $H_1$ is fully connected, i.e., $b_{k\ell}\neq 0$ except
            (possibly) for $k=\ell$.
\end{itemize}
\end{definition}

\begin{theorem}
\label{thm:lasalle:3} The subspace $\B^{n^2-n}$ generated by the
Ad-brackets is equal to $\T$ if and only if the Hamiltonian is ideal.
\end{theorem}

\begin{proof}
The technique in this proof is a direct application of the property
of Vandermonde matrix, which has also been applied to discuss the
controllability~\cite{Leite,altafini-root}. Since the dimension of
$\T$ is $n^2-n$ and $B_m\in\T$ for all $m>0$, it suffices to show
that the elements $B_m$ for $m=1,\ldots,n^2-n$ are linearly
independent.  Moreover, the subspaces spanned by the odd and even
order elements, $\B_{\rm odd}^s=\Span\{B_{2m-1}: 1\le 2m-1\le s\}$
and $\B_{\rm even}^s = \Span\{B_{2m}: 1\le 2m\le s\}$, respectively,
are orthogonal since
\begin{multline*}
 B_{2m-1} B_{2m'}
 = (-1)^{m+m'}\sum_{k,\ell}\sum_{k',\ell'}
   \omega_{k\ell}^{2m-1}
   \omega_{k'\ell'}^{2m'} \times \\
  [\Re(b_{k\ell}) \Re(b_{k'\ell'}) \bar{\lambda}_{k\ell} \lambda_{k'\ell'}
  -\Im(b_{k\ell}) \Im(b_{k'\ell'}) \lambda_{k\ell} \bar{\lambda}_{k'\ell'}\\
  -\Re(b_{k\ell}) \Im(b_{k'\ell'}) \bar{\lambda}_{k\ell}\bar\lambda_{k'\ell'}
  +\Im(b_{k\ell}) \Re(b_{k'\ell'}) \lambda_{k\ell} \lambda_{k'\ell'}],
\end{multline*}
and thus observing the equalities
\begin{subequations}
\begin{gather}
  \Tr(\lambda_{k\ell}\lambda_{k'\ell'})
  = \Tr(\bar\lambda_{k\ell}\bar\lambda_{k'\ell'})
  = -2\delta_{kk'}\delta_{\ell\ell'} \\
   \Tr(\lambda_{k\ell}\bar\lambda_{k'\ell'}) =0
\end{gather}
\end{subequations}
shows that for all $m,m'>0$
\begin{multline*}
 \Tr(B_{2m-1} B_{2m'})
 = (-1)^{m+m'}\sum_{k,\ell}\sum_{k',\ell'}
   \omega_{k\ell}^{2m-1} \omega_{k'\ell'}^{2m'} \times \\
   \Re(b_{k\ell}) \Im(b_{k\ell})[-\bar{\lambda}_{k\ell}^2+\lambda_{k\ell}^2]=0.
\end{multline*}
Thus it suffices to show that the elements of $\B_{\rm odd}^{n^2-n}$
and $\B_{\rm even}^{n^2-n}$ are linearly independent separately.

For the odd terms, suppose there exists a vector
$\vec{c}=(c_1,\ldots,c_s)^T$ of length $s=n(n-1)/2$ such that
$\sum_{m=1}^s c_m B_{2m-1}=0$.  Noting that $\omega_{kk}=0$ and
$(\omega_{\ell k})^2=(-\omega_{k\ell})^2$ this gives $n(n-1)/2$
non-trivial equations
\begin{equation}
\label{eq:lin_indep}
  \omega_{k\ell} [\Re(b_{k\ell}) \bar{\lambda}_{k\ell}
                        +\Im(b_{k\ell}) \lambda_{k\ell}] \;
  \sum_{m=1}^{s} (-\omega_{k\ell}^2)^{m-1} c_m=0,
\end{equation}
for $1\le k<\ell\le n$.  Since $\omega_{k\ell}\neq 0$, $b_{k\ell} \neq0$
by hypothesis, Eq.~(\ref{eq:lin_indep}) can be reduced to $\Omega
\vec{c}=\vec{0}$, where $\Omega$ is a matrix:
\begin{equation}
  \label{eq:vandermonde}
  \Omega = \begin{bmatrix}
    1 & -\omega_{12}^2 & \omega_{12}^4 & \ldots & (-\omega_{12}^2)^{m-1} \\
    1 & -\omega_{13}^2 & \omega_{13}^4 & \ldots & (-\omega_{13}^2)^{m-1} \\
    \vdots & \vdots & \vdots & \ddots & \vdots \\
    1 & -\omega_{n-1,n}^2 & \omega_{n-1,n}^4 & \ldots &
(-\omega_{n-1,n}^2)^{m-1} \\
      \end{bmatrix}.
\end{equation}
Since $\Omega$ is a Vandermonde matrix, condition~(ii) of the
proposition guarantees that Eq.~(\ref{eq:lin_indep}) has only the
trivial solution $\vec{c}=\vec{0}$, thus establishing linear
independence.  For the even terms we obtain a similar system of
equations, which completes the proof.
\end{proof}

\begin{theorem}
\label{thm:lasalle:4} Assuming the Hamiltonian of (\ref{eqn:auto})
is ideal, $(\rho_1,\rho_2)$ belongs to the invariant set $E$ if and
only if $[\rho_1,\rho_2]=\diag(c_1,\ldots,c_n)$.
\end{theorem}

\begin{proof}
The necessary part follows from Eq.~(\ref{eq:trace-cond1}) and
Theorem~\ref{thm:lasalle:3}.  By (\ref{eq:trace-cond1}) the states in
the invariant set have to satisfy $\Tr(B[\rho,\rho_d])=0$ for all $B\in
\mathcal{B}^{n^2-n}$, and by Theorem III.4 we have $\mathcal{B}^{n^2-n}
=\T$ for ideal systems.  Thus, $[\rho,\rho_d]$ must be in the subspace
orthogonal to $\T$, which is the Cartan subspace, i.e., the diagonal
(trace-zero) matrices.

For the sufficient part note that 
$\rho_k(t)=e^{-iH_0t}\rho_k e^{iH_0t}$, $k=1,2$,
\[
 [e^{-iH_0t}\rho_1 e^{iH_0t},e^{-iH_0t}\rho_2e^{iH_0t}]
  = e^{-iH_0t}[\rho_1,\rho_2] e^{iH_0t}
\]
and $e^{-iH_0 t}$ diagonal.  
Thus if $[\rho_1,\rho_2]=\diag(c_1,\ldots,c_n)$ then 
\(
  e^{-iH_0t}[\rho_1,\rho_2] e^{iH_0t} =
 \diag(c_1,\ldots,c_n) = [\rho_1,\rho_2]
\) and hence $(\rho_1,\rho_2)\in E$.
\end{proof}

Thus we have fully characterized the invariant set for systems with
ideal Hamiltonians.  The result also shows that even under the most
stringent assumptions on the Hamiltonian, the invariant set is generally
much larger than the desired solution.  Therefore, the invariance
principle alone is not sufficient to establish convergence to the target
state.

\begin{remark}
Theorem~\ref{thm:lasalle:3} shows that for a system with ideal
Hamiltonian the $\Ad$-brackets span the entire tangent space to the
state manifold, $\B^{n^2-n}=\T$, and thus the linearization defined on
the tangent space is controllable.  It is also easy to verify that this 
condition is necessary for controllability of the linearization for 
generic states.
\end{remark}

\section{Lyapunov Function as Morse Function}
\label{sec:critical}

We show the LaSalle invariant set of (\ref{eqn:auto}) always contains
the critical points of the Lyapunov function $V$, and characterize the
stability of these critical points.  Notice that the stability here
refers to the stability of these critical points as stationary states of
the gradient flow induced by the Lyapunov function.  The gradient flow
in general need not be related to the particular dynamics of the system,
e.g., prescribed by the equation of motion~(\ref{eqn:auto}).  Indeed we
shall see that (\ref{eqn:auto}) is not a gradient flow.  The stability
analysis of the critical points of $V$ is still useful, however, because
under certain conditions, for instance, when the stationary state of
(\ref{eqn:auto}) is hyperbolic, the dimensions of the stable (unstable)
manifold of (\ref{eqn:auto}) agree with those of the stable (unstable)
manifold of the gradient flow.

We start with the case where $\rho_d$ is given, and $V(\rho,\rho_d)$ is
effectively a function of $\rho$ on $\M$.  Similar topics have also been
discussed in~\cite{Wu}, but the major point here is that we can find the
critical points of $V(\rho_1,\rho_2)$ on $\M\times\M$ from the critical
points of $V(\rho)$ on $\M$.  Since $\rho$ can be written as $\rho=
U\rho_d U^\dagger$ for some $U$ in the special unitary group $\SU(n)$,
$V$ can also be considered as a function on $\SU(n)$,
$V(U)=V(U\rho_dU^\dag\rho_d)$.  It is easy to see that the critical
points of $V(\rho)$ correspond to those of $V(U)$, and as
$\Tr(\rho_d^2)$ is a constant for a given $\rho_d$, it is equivalent to
find the critical points $U\in\SU(n)$ of
\begin{align}
  \label{eq:J}
  J(U) =\Tr (\rho_d^2)-V(U) = \Tr(U \rho_d U^\dagger \rho_d).
\end{align}

\begin{lemma}
\label{lemma:crit:1} The critical points $U_0$ of $J(U)$ defined
by~(\ref{eq:J}) are such that $[\rho_0,\rho_d]=0$ for
$\rho_0=U_0\rho_d U_0^\dagger$.
\end{lemma}

\begin{IEEEproof}
Let $\{\sigma_m\}_{m=1}^{n^2-1}$ be the orthonormal basis of $\su(n)$
given in Appendix \ref{app:lie_algebra_basis}.  Any $U\in\SU(n)$ near
the identity $I$ can be written as $U=e^{\vec x\cdot\vec\sigma}$, where
$\vec{\sigma}=(\sigma_1,\ldots,\sigma_{n^2-1})$ and $\vec{x}\in\RR^n$ is
the coordinate of $U$, and any $U$ in the neighborhood of $U_0$ can be
parameterized as $U=e^{\vec{x}\cdot\vec\sigma}U_0$.  Thus (\ref{eq:J})
becomes
\begin{equation}
\label{eq:J:local} 
J = \Tr[ (e^{\xs} U_0) \rho_d (U_0^\dagger e^{-\xs})  \rho_d].
\end{equation}
At the critical point $U_0$, $\nabla J=0$ implies
\begin{align}
 0= \Tr(\sigma_m[U_0\rho_d U_0^\dagger, \rho_d]) \quad \forall m.
\end{align}
Thus $[U_0\rho_d U_0^\dagger,\rho_d]\in\su(n)$ is orthogonal to all
basis elements $\sigma_m$, and therefore $[U_0\rho_d
U_0^\dagger,\rho_d]=0$.
\end{IEEEproof}

More generally, for $V(\rho_1,\rho_2)$ defined on $\M\times\M$, with
$\rho_1 = U_1\rho_d U_1^\dag$ and $\rho_2 = U_2\rho_dU_2^\dag$, the
critical points of $J(\rho_1,\rho_2)=\Tr(\rho_1\rho_2)$ and
$V(\rho_1,\rho_2)=\Tr(\rho^2)-\Tr(\rho_1\rho_2)$ coincide as
$\Tr(\rho_2^2)=\Tr(\rho_d^2)$ is constant.
\begin{align*}
J(U_1,U_2) &=\Tr\Big( U_1\rho_d U_1^\dag U_2\rho_d U_2^\dag\Big)\\
           &=\Tr\Big((U_2^\dag U_1)\rho_d(U_2^\dag U_1)^\dag\rho_d\Big)
\end{align*}
together with Lemma~\ref{lemma:crit:1} shows that $J$ attains its
critical value when $[(U_2^\dag U_1)\rho_d(U_2^\dag U_1)^\dag,
\rho_d]=0$ and
\begin{align*}
 0 &= U_2 [(U_2^\dagger U_1)\rho_d(U_2^\dagger U_1)^\dagger,
\rho_d]U_2^\dagger\\
   &= [U_1\rho_d U_1^\dagger,U_2 \rho_d U_2^\dagger] = [\rho_1,\rho_2].
\end{align*}
Thus we have the following:
\begin{theorem}
\label{thm:crit:1} For a given target state $\rho_d(0)$, the
critical points of the Lyapunov function $V(\rho_1,\rho_2)$ on
$\M\times\M$ satisfy $[\rho_1,\rho_2]=0$ and thus belong to the
LaSalle invariant set $E$.
\end{theorem}

Next, for a given generic $\rho_d$, $J(\rho)=\Tr(\rho\rho_d)$, and
thus $V(\rho)=V(\rho,\rho_d)$, are Morse functions on $\M$, i.e. its
critical points are hyperbolic~\cite{Matsumoto}:
\begin{theorem}
\label{thm:crit:2} If $\rho_d$ is generic then
$J(\rho)=\Tr(\rho\rho_d)$ is a Morse function on $\M$. Two of the
$n!$ hyperbolic critical points correspond to the global maximum and
minimum of $J$, respectively, and the other $n!-2$ points are
saddles with critical values $J_0$ satisfying $J_{\rm min}< J_0 <
J_{\rm max}$.
\end{theorem}

\begin{IEEEproof}
See \cite{stability} for the definition of sink, source and saddles.
For a given $\rho_d$, $\rho_0$ is a critical point of $J(\rho)$ if and
only if $[\rho_0,\rho_d]=0$, i.e. there exists a basis such that
\begin{align*}
  \rho_d &=\diag(w_1,\ldots,w_n), \\
 \rho_0 &=\diag(w_{\tau(1)},\ldots,w_{\tau(n)}),
\end{align*}
where $\tau$ is a permutation of the numbers $\{1,\ldots,n\}$. Since
$\rho_d$ is generic and hence the $w_k$ are distinct, there are $n!$
distinct permutations and thus $n!$ critical points with critical
values $J(\rho_0)=\sum_{k=1}^n w_k w_{\tau(k)}$.

Next, in order to calculate the Hessian matrix, we need to find a
parameterization of the points near $\rho_0$. Recalling that for any
$\rho\in\M$ we have $\rho=U\rho_d U^\dag$ for some $U\in\SU(n)$,
consider $J$ as a function on $\SU(n)$ with
$J(U)=\Tr(U\rho_dU^\dagger\rho_d)$.  Let $U_0$ be a critical point
of $J(U)$ and $\rho_0=U_0\rho_d U_0^\dag$.  Any $U$ in the
neighborhood of $U_0$ can be parameterized as $U=e^{\xs}U_0$, where
$\vec{x}\in\RR^{n^2-1}$ and $\vec{\sigma}$ is the orthonormal basis
for $\su(n)$ defined in appendix \ref{app:lie_algebra_basis}.  
Substituting this into $J$ gives
\begin{align*}
J(\vec{x}) 
 =& \Tr(e^{\xs}U_0\,\rho_d\,U_0^\dagger e^{-\xs}\rho_d)\\
 =& \Tr[(\ONE+\xs+\mbox{$\frac{1}{2}$}(\xs)^2) \, U_0\rho_d U_0^\dagger\times\\
  & \qquad (\ONE-\xs+\mbox{$\frac{1}{2}$}(\xs)^2)\rho_d]+\Theta(|\vec x|^3)\\
 =& \Tr[U_0\,\rho_d\, U_0^\dagger\, \rho_d]
    +\Tr[(\xs)^2 \rho_0 \rho_d]\\
  & -\Tr[(\xs)\rho_0(\xs)\rho_d]
    +\Theta(|\vec{x}|^3),
\end{align*}
where we have used $\rho_0 \rho_d=\rho_d \rho_0$.  Taking the basis to
be $\vec{\sigma}=\{\lambda_{k}, \lambda_{k\ell}, \bar\lambda_{k\ell}\}$
with $\lambda_k$, $\lambda_{k\ell}$ and $\bar\lambda_{k,\ell}$ as in
appendix \ref{app:lie_algebra_basis}, we find that the Hessian matrix
$\frac{\partial^2 J}{\partial{x_j}\partial{x_j}}$ at $\rho_0$ is
diagonal, i.e., the basis vectors are eigenvectors.  The first $n-1$
diagonal entries corresponding to $\lambda_k$ vanish but as we are only
interested in the tangent space to the manifold spanned by
$\{\lambda_{k\ell},\bar{\lambda}_{k\ell}\}$, we can restrict our
attention to this subspace.  On this subspace, i.e., for
$\sigma_j=\lambda_{k\ell}$ or $\bar\lambda_{k\ell}$, we have
\begin{align*}
\frac{\partial^2 J}{\partial^2{x_j}}
=2\Tr[\sigma_j^2\rho_0\rho]-2\Tr[\sigma_j\rho_0\sigma_j\rho].
\end{align*}
The action of $\sigma_j=\lambda_{k\ell}$ or
$\sigma_j=\bar\lambda_{k\ell}$ is restricted to the subspace spanned by
the basis vectors $e_k$ and $e_{\ell}$.  On this subspace
$\lambda_{k\ell}^2$ is identity operator and the conjugate action of
$\sigma_j$ on the diagonal matrix $\rho_0$ swaps its $k$-th and
$\ell$-th diagonal entries. Since $\rho_0$ is non-degenerate, any swap
$\lambda_{k\ell}$ or $\bar\lambda_{k\ell}$ will make $\frac{\partial^2
J}{\partial^2{x_j}}$ either larger or smaller than zero. Thus the
Hessian matrix at $\rho_0$ is diagonal with $n^2-n$ non-zero diagonal
entries, corresponding to $n^2-n$ independent directions in the tangent
space of $\M$. Therefore, all $n!$ critical points $\rho_0$ are
hyperbolic, and $J$ is a Morse function.  The maximal critical value
occurs only when $\rho_0=\rho_d$ and the minimal value occurs only when
$w_{\tau(k)}$'s are in an increasing order. For all other critical
values, there always exists a swap that will increase the value of $J$
and one that will decrease it, showing that they are saddle points.
\end{IEEEproof}

\section{Effectiveness of Lyapunov control for ideal systems}
\label{sec:conv_ideal}

When the hyperbolic critical points of the Lyapunov function are
also the stationary points of the dynamics, there are restrictions
on the possible dynamics near those critical points.  In particular,
if the dynamics is the gradient flow of the the Lyapunov function
then there is a simple correspondence between the number of negative
(positive) eigenvalues at the critical point and the dimension of
the stable (unstable) manifold at the critical point as a stationary
solution.  However, in general, this does not hold for a dynamical
system other than the gradient flow.  To be a gradient flow, the
coefficient matrix of the linearized system has to be symmetric and
we will see that (\ref{eqn:auto}) is not the gradient flow of any
function.  Therefore, in order to investigate the stability and to
calculate the dimension of the stable manifold at any stationary
point, we have to resort to the definition of the stable manifold,
and investigate the linearized dynamics.

Throughout this section we shall assume that the Hamiltonian is
\emph{ideal}.  Without loss of generality we further assume that $H_0$
has zero trace, as the identity part of $H_0$ only changes the global
phase.  Once the Hamiltonian is chosen, the LaSalle invariant set $E$
depends only on the target state $\rho_d$.  As stated before, throughout
this paper we focus on generic states, i.e., assuming $\rho_d$ has $n$
distinct eigenvalues, and assume $\rho(0)$ and $\rho_d(0)$ have the same
spectrum.  Similar tools can be applied to non-generic states but they
must be separately investigated as the topology of the critical points
for non-generic states is different.

\subsection{Stationary (generic) target state}

We work in a basis where $H_0$ is diagonal.  If $\rho_d$ is stationary,
i.e., $[H_0,\rho_d]=0$, then it is also diagonal, and (\ref{eqn:auto})
reduces to a dynamical system on $\M$:
\begin{subequations}
\label{eqn:auto1}
\begin{align}
\dot \rho(t) &=-i [ H_0+f(\rho)H_1, \rho(t) ]\\
f(\rho)&=\Tr([-iH_1,\rho(t)]\rho_d)
\end{align}
\end{subequations}
with the corresponding LaSalle invariant set
\begin{align}
  E &=\{\rho_0: \dot V(\rho(t))=0,\rho(0)=\rho_0\} \nonumber\\
    &=\{\rho_0: [\rho_0,\rho_d]=\diag(c_1,\ldots,c_n)\}
\end{align}
Let $\rho_d=\diag(w_1,\ldots,w_n)$, with $w_k\neq w_\ell$ for
$k\ne\ell$. For any $\rho\in E$, $[\rho_d,\rho]$ is diagonal if and
only if $\rho$ is diagonal, with diagonal elements as a permutation
of $(w_1,\ldots,w_n)$. According to previous section, these $n!$
stationary points are also the hyperbolic critical points of the
Lyapunov function $V(\rho)$.

\begin{theorem}
\label{thm:generic:crit} If $\rho_d$ is a generic stationary state then
the invariant set contains exactly the $n!$ critical points of the
Lyapunov function $\rho_d^{(k)}$, $k=1,\ldots,n!$, that commute with
$\rho_d$ and have the same spectrum.
\end{theorem}

These $n!$ points are the only stationary solutions and all the
other solutions must converge to \emph{one} of these points. Through
analyzing the sign of the eigenvalues of the coefficient matrix of
the linearized system, we shall see that $\rho_d$ is asymptotically
stable, and all other stationary points are unstable. In order to
achieve this, we require a real representation for
(\ref{eqn:auto1}). A natural choice is the Bloch representation. Let
$\{\xi_k\}_{k=1}^{n^2}$ be an orthonormal basis for all $n\times n$
Hermitian matrices, with $\xi_{n^2}=\frac{1}{\sqrt{n}}I$. We have
$\rho=\sum_k s_k\xi_k$, with $s_k=\Tr(\rho \xi_k)$. Since the
dynamics is trace-preserving, i.e. $s_{n^2}=\frac{1}{\sqrt{n}}$ is
constant, we can further reduce the dynamics onto the subspace
$\RR^{n^2-1}$, and $\rho$ can be represented as a vector
$\vec{s}\in\RR^{n^2-1}$. Accordingly, the quantum dynamical
system~(\ref{eqn:auto}) can be represented as
\begin{subequations}
\begin{align*}
\dot {\vec{s}}(t)   &= (A_0+f(\vec{s},\vec{s}_d)A_1)\vec{s}(t)\\
\dot {\vec{s}}_d(t) &= A_0\vec{s}_d(t)\\
f(\vec{s},\vec{s}_d)&= \vec{s_d}^TA_1\vec{s},
\end{align*}
\end{subequations}
where $A_0$ and $A_1$ are two anti-symmetric matrices:
\begin{align}
A_0 (m,n)  &= \Tr(iH_0[\xi_m,\xi_n])\\
A_1 (m,n)  &= \Tr(iH_1[\xi_m,\xi_n])
\end{align}
When $\rho_d$ is stationary, this system reduces to
\begin{subequations}
\label{eq:sys_real}
\begin{align}
\dot {\vec{s}}(t) &= (A_0+f(\vec{s})A_1)\vec{s}(t)\\
        f(\vec{s})&= \vec{s_d}^TA_1\vec{s},
\end{align}
\end{subequations}
and the Lyapunov function (\ref{eqn:4}) is represented as
$V(\vec{s})=\frac{1}{2}||\vec{s}-\vec{s}_d||^2$. According to
Theorem~\ref{thm:generic:crit}, for a generic $\rho_d$,
(\ref{eq:sys_real}) has $n!$ stationary points, denoted as ${\vec
s}^{(k)}$, $k=1,\ldots,n!$. The linearized system near the
stationary state ${\vec s}^{(k)}$ is
\begin{equation}
\label{eqn:linear}
  \dot {\vec{s}}= D_f({\vec s}^{(k)})\cdot (\vec{s}-{\vec s}^{(k)}),
\end{equation}
where $D_f({\vec s}^{(k)})=A_0+\vec{s_d}^TA_1{\vec s}^{(k)}A_1+A_1
\vec{s}^{(k)}\cdot \vec{s_d}^T A_1$ is a linear map defined on
$\RR^{n^2-1}$. $f({\vec s}^{(k)})=0$ gives $\vec{s_d}^TA_1{\vec
s}^{(k)}=0$, and $D_f({\vec s}^{(k)})=A_0+A_1 \vec{s}^{(k)}\cdot
\vec{s_d}^T A_1$. Since $A_0$ and $A_1$ are anti-symmetric,
$D_f({\vec s}^{(k)})$ cannot be a symmetric matrix, and the dynamics
cannot be a gradient flow of any function. Therefore, the topology
near $\vec{s}^{(k)}$ as a critical point of $V$ is not enough to
infer the local dynamics in its vicinity, and we need to actually
calculate the eigenvalues of $D_f(\vec{s}^{(k)})$.

\begin{remark}
The state space $S_\M$ of (\ref{eq:sys_real}) is the set of all Bloch
vectors $\vec{s}\in \RR^{n^2-1}$ corresponding to density operators
$\rho\in\M$.  For generic $\rho_d$ the state manifold $\M$ is a flag
manifold homeomorphic to $\SU(n)/\exp(\C)$, where $\C$ is the Cartan
subspace of the Lie algebra $\su(n)$ and $\exp(\C)$ is its exponential
image in $\SU(n)$, corresponding to diagonal unitary matrices with
determinant $1$.  Hence, the tangent space $T_\M(\rho_0)$ of $\M$ at any
point $\rho_0$ corresponds to the non-Cartan subspace $\T$ of $\su(n)$
and the Cartan elements of $\su(n)$ correspond to the tangent space of
the isotropy subgroup of $\rho_0$.  In the real representation,
$\RR^{n^2-1}$ is therefore the direct sum of the ($n^2-n$)-dimensional
tangent space $S_\T$ to the manifold $S_\M$ and the ($n-1$)-dimensional
subspace $S_\C$ corresponding to the Cartan subspace of $\su(n)$.
\end{remark}

\begin{theorem}
\label{thm:generic:hyperbolic} For a generic stationary target state
$\rho_d$ all the $n!$ stationary states of the dynamical
system~(\ref{eqn:auto1}) are hyperbolic, i.e. all eigenvalues of
$D_f({\vec s}^{(k)})$, restricted on $S_\T$, have nonzero real
parts, for $k=1,\ldots,n!$. Among those stationary states, $\rho_d$
is the only sink, all other points are saddles, except the global
maximum, which is a source.
\end{theorem}

\begin{IEEEproof}
Let $\vec{s}_0$ be one of the $n!$ stationary states. We first show
that $D_f(\vec{s}_0)$ vanishes on the $(n-1)$-dimensional subspace
$S_\C$, which is orthogonal to $S_\T$. In the second step we show
that $D_f(\vec{s}_0)$ is invariant on $S_\T$ and has $n^2-n$
non-zero eigenvalues. Finally, we show that the restriction of
$D_f(\vec{s}_0)$ onto $S_\T$ does not have any purely imaginary
eigenvalues, from which it follows that $\vec{s}_0$ is a hyperbolic
stationary state, and the local dynamics of (\ref{eqn:auto1}) near
every stationary state can therefore be approximated by the
linearized system~\cite{stability}.

\begin{lemma}
$D_f(\vec{s}_0)$ vanishes on the subspace $S_\C$.
\end{lemma}

This lemma shows that $\vec{s}_0$ is not a hyperbolic fixed point of
the dynamical system~(\ref{eq:sys_real}) defined on $\RR^{n^2-1}$.
However, we are only interested in the dynamics on the manifold
$S_\M$, and thus it suffices to show that $\vec{s}_0$ is a
hyperbolic fixed point of the restriction of $D_f(\vec{s}_0)$ to the
tangent space $S_\T$ of $S_\M$.

\begin{lemma}
The restriction of $D_f(\vec{s}_0)$ to $S_\T$ is well-defined,
represented by a matrix $B$ with $n^2-n$ non-zero eigenvalues.
\end{lemma}

\begin{lemma}
\label{lemma:generic1:3} If $i\beta$ is a purely imaginary
eigenvalue of $B$ then it must be an eigenvalue of $B_0$, i.e.,
$i\beta=\pm i\omega_{k\ell}$ for some $(k,\ell)$, and either the
associated eigenvector $\vec{e}$ must be an eigenvector of $B_0$
with the same eigenvalue, or the restriction of $A_1\vec{s}_0$ to
the $(k,\ell)$ subspace must vanish.
\end{lemma}

Lemma~\ref{lemma:generic1:3} shows that $B$ can have a purely imaginary
eigenvalue $i\beta$ only if $i\beta=\pm i\omega_{k\ell}$ for some
$(k,\ell)$, and either $\vec{u}^{(k,\ell)}=\vec{0}$, i.e., the
projection of $A_1\vec{s}_0$ onto the $(k,\ell)$ subspace vanishes, or
the associated eigenvector is also an eigenvector of $B_0$.  In the
first case this means that $A_1\vec{s}_0$ vanishes on the subspace
$\T_{k\ell}$, or equivalently that $[-iH_1,\rho_0]$ has no support in
$\T_{k\ell}$, which contradicts the assumption that $H_1$ is fully
connected and $\rho_0$ has non-degenerate eigenvalues.  On the other
hand, if $\vec{e}$ is an eigenvector of $B_0$ with eigenvalue
$i\beta=\pm i\omega_{k\ell}$ and $H_0$ is strongly regular then the
projection of $\vec{e}$ onto the $(k,\ell)$ subspace is proportional to
$(1,\pm i)$ and $\vec{e}$ is zero elsewhere, and thus
$\vec{v}^T\vec{e}=0$ implies $\vec{v}^{(k,\ell)}=0$, which contradicts
the fact that the projection $A_1\vec{s}_d$ or $[-iH_1,\rho_d]$ onto the
$(k,\ell)$ subspace must not vanish if $H_1$ is fully connected and
$\rho_d$ has non-degenerate eigenvalues.  Thus we can conclude that if
$H_0$ is strongly regular, $H_1$ fully connected and $\rho_d$ has
non-degenerate eigenvalues, $D_f(\vec{s}_0)$ cannot have purely
imaginary eigenvalues, and thus $\vec{s}_0$ is hyperbolic.
\end{IEEEproof}

This theorem illustrates that the $n!$ critical points $\rho^{(k)}$ of
$V$ are also the $n!$ hyperbolic stationary states of (\ref{eqn:auto1}).
Since $V(\rho_d)=0$, $\rho=\rho_d$ must be a dynamical sink, with all
eigenvalues of $D_f(\vec{s}_d)$ having negative real parts. Similarly,
$\rho=\rho^{(n)}$ with $V(\rho^{(n)})=V_{\rm max}$ must be a dynamical
source. All the other $\rho^{(k)}$ with $0<V(\rho^{(k)})<V_{\rm max}$
must be saddles, with eigenvalues of $D_f(\vec{s}^{(k)})$ having both
negative and positive real parts, for otherwise $\rho^{(k)}$ would be a
sink or source, and thus a local minimum or maximum of $V$, which would
contradict Theorem \ref{thm:crit:2}. Moreover, the dimension of the
stable (unstable) manifold at $\rho^{(k)}$ must agree with the the index
number of $V(\rho^{(k)})$ at $\rho^{(k)}$.  This is a very useful
observation as it allows us to infer that the dimension of the stable
(unstable) manifold at $\rho^{(k)}$ is independent of the specific value
of $\rho_d$, only dependent on the relative location of $\rho^{(k)}$ as
a critical point of $V$ and the system dimension $n$.  The theorem also
shows that each of the $n!-2$ saddle points, $\rho^{(k)}$ has a stable
manifold.  Solutions on the stable manifold converge to $\rho^{(k)}$ and
thus the saddles are not repulsive, as asserted in \cite{altafini2}, and
we can in construct counter-examples to Theorem~1 in \cite{altafini2}.

For example, consider a three-level system with $H_0$ strongly regular
and $H_1$ off-diagonal and fully connected.  For a generic stationary
target state such as $\rho_d=\frac{1}{6}\diag(3,2,1)$, the LaSalle
invariant set consists of $3!=6$ stationary states---$\rho^{(1)}=\rho_d$
and five other $\rho^{(k)}$ referred to as the antipodal points
in~\cite{altafini2}.  The coefficient matrix $D_f(\rho^{(k)})$ of the
linearized system has eigenvalues with negative real parts for every
$\rho^{(k)}$ except the global maximum
$\rho^{(6)}=\frac{1}{6}\diag(1,2,3)$, and thus four of the antipodal
points have stable manifolds and solutions converging to them.  An even
easier way to see that these points cannot all be repulsive is to note
that if, e.g., $\rho_0=\frac{1}{6}\diag(2,3,1)$ was repulsive then we
would have $V(\rho(t))\le V(\rho_0)$ for all $\rho(t)$ in a neighborhood
of $\rho_0$, and thus $\rho_0$ would be a local maximum of the Lyapunov
$V(\rho)=\frac{1}{2}\Tr(\rho-\rho_d)^2$, contradicting the fact that it
is a saddle point of $V$.  On the other hand, we note that any state
$\rho(t)$ starting outside the invariant set $E$ has at least one
off-diagonal component, and as $H_1$ is fully connected and $\rho_d$
non-degenerate, the off-diagonal components of $[-iH_1,\rho_d]$ are all
nonzero.  Thus, the trajectories converging to the saddle points satisfy
conditions (1) and (2) of Theorem~1 in~\cite{altafini2}.  Condition~(3)
is also satisfied as $\mbox{\rm Card}\F_t([-iH_1,\rho_d])=3= \dim\M/2$,
where $\dim\M=3^2-3$.  Thus by Theorem~1 in~\cite{altafini2} they
should converge to $\rho_d$, which is not the case.

Nonetheless, the stable manifolds of the unstable stationary states are
not a serious obstruction to convergence.  In fact, since all solutions
not converging to $\rho_d$ are located on the union of the $n!-2$ stable
manifolds of dimension $<n^2-n$, which form a measure-zero set in the
state space, we can conclude that almost all initial states converge to
$\rho_d$, i.e., $\rho_d$ can be considered almost globally
asymptotically stable, and the Lyapunov design effective in this case.
This is illustrated in Fig.~\ref{fig1}(a) which shows that for a
stationary generic $\rho_d$, all simulated non-stationary trajectories
with random $\rho(0)$ converge to $\rho_d$ exponentially.

\begin{figure*}
\caption{Time evolution of $V(\rho(t),\rho_d(t))$ with $y$-axis in
logarithmic scale.  Each graph shows $V(\rho(t),\rho_d(t))$ for $N=50$
different initial states $\rho(0)$.  The graphs represent four different
types of generic $\rho_d$, of which (a,b,d) are for ideal Hamiltonian,
and (c) for non-ideal Hamiltonian. (a) shows that for stationary
$\rho_d$, all trajectories converge exponentially to the target state to
within machine precision.  The negative slopes in (b) suggest that for a
non-stationary target state with regular $E$, all simulated trajectories
still converge to the target trajectory albeit at a slower rate compared
to (a).  For a non-stationary $\rho_d$ with irregular $E$ as in (d), or
a stationary $\rho_d$ with $H_1$ not fully connected as in (c), on the
other hand, the slopes of $V(\rho(t),\rho_d(t))$ in the log-plot vanish
at different finite distances from the target state for all simulated
trajectories, indicating convergence to states or trajectories at
various non-zero distances from the target.}  \label{fig1}
\subfloat[Stationary target state]%
  {\includegraphics[width=0.49\textwidth]{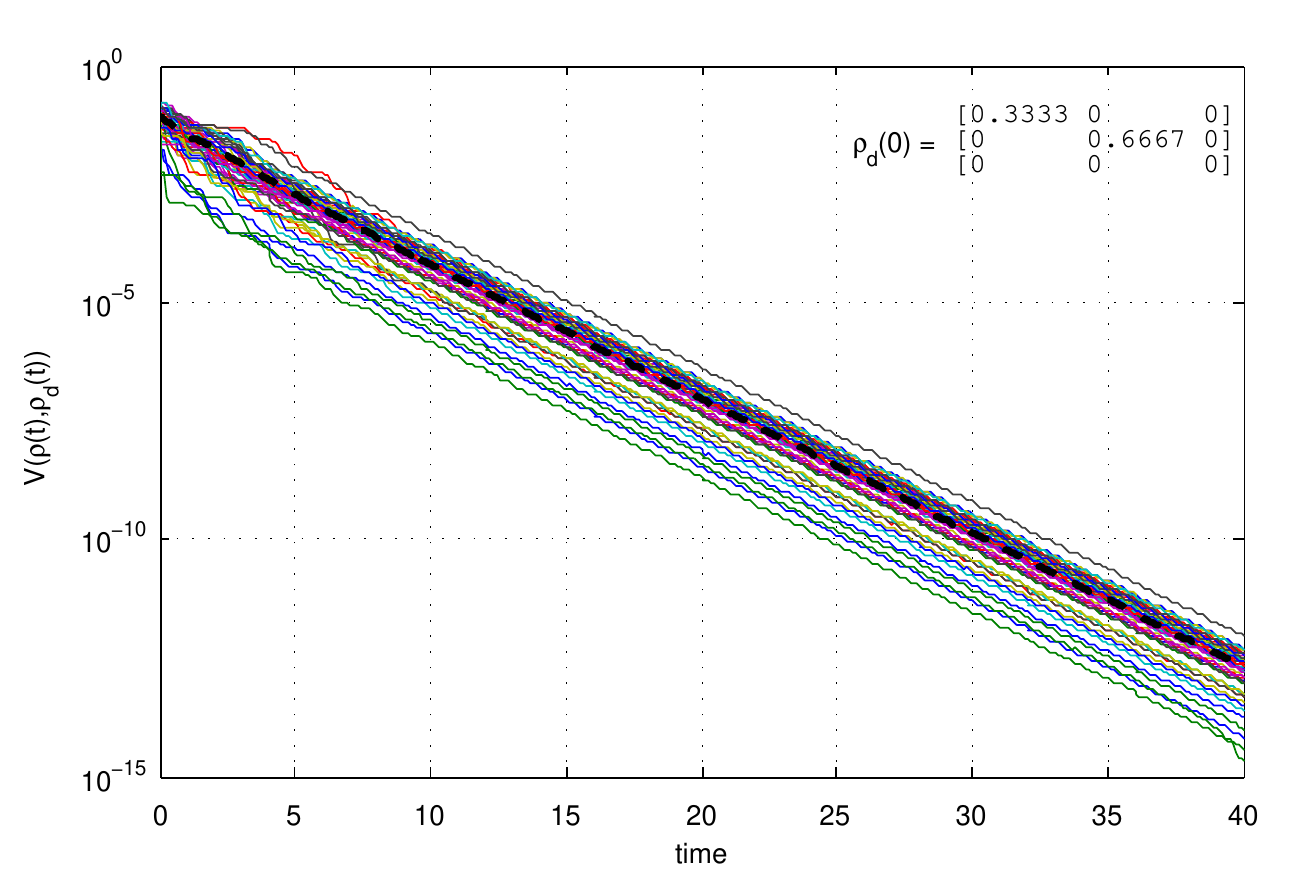}}
\subfloat[Non-stationary target state with regular $E$]%
  {\includegraphics[width=0.49\textwidth]{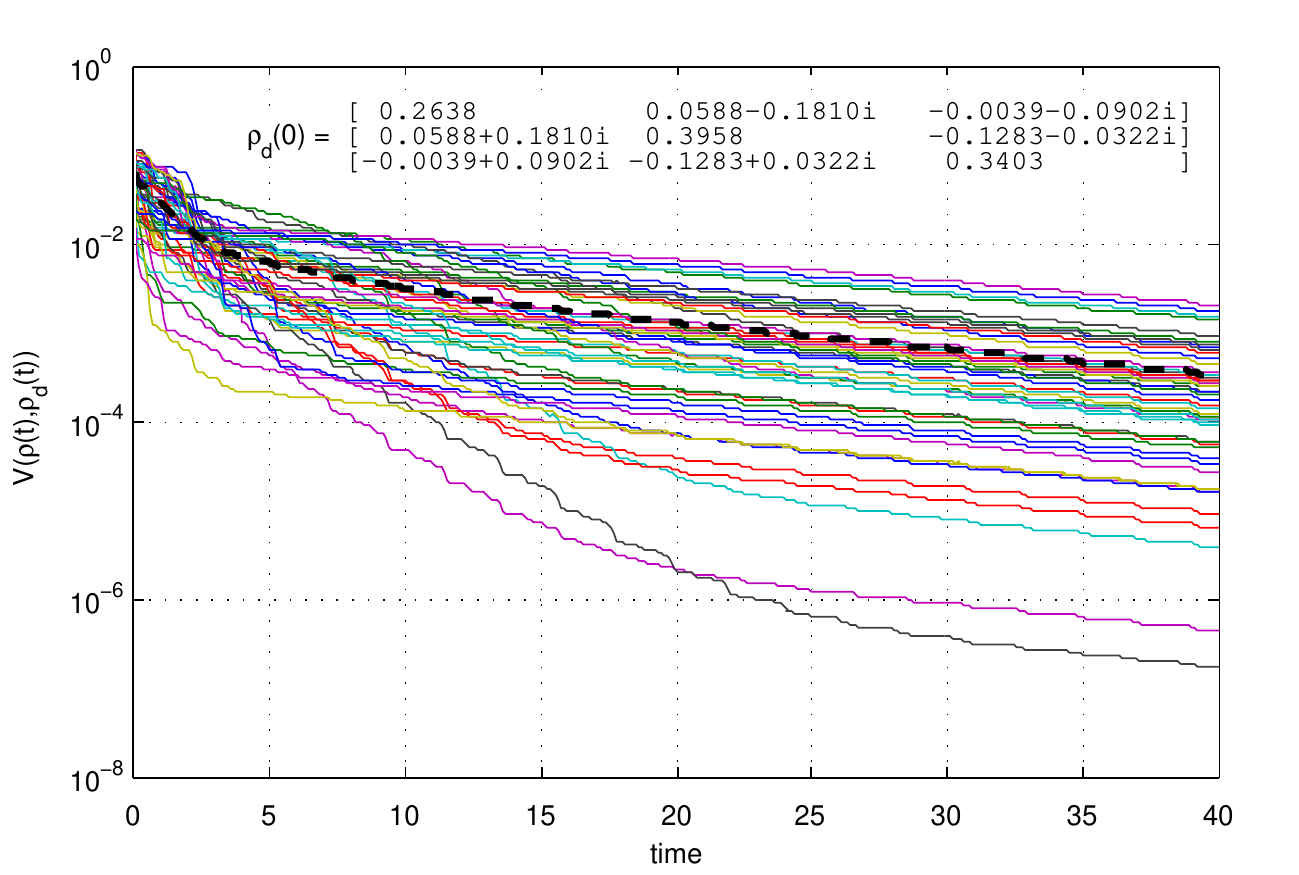}}\\
\subfloat[Stationary target state with $H_1$ not fully connected]%
  {\includegraphics[width=0.49\textwidth]{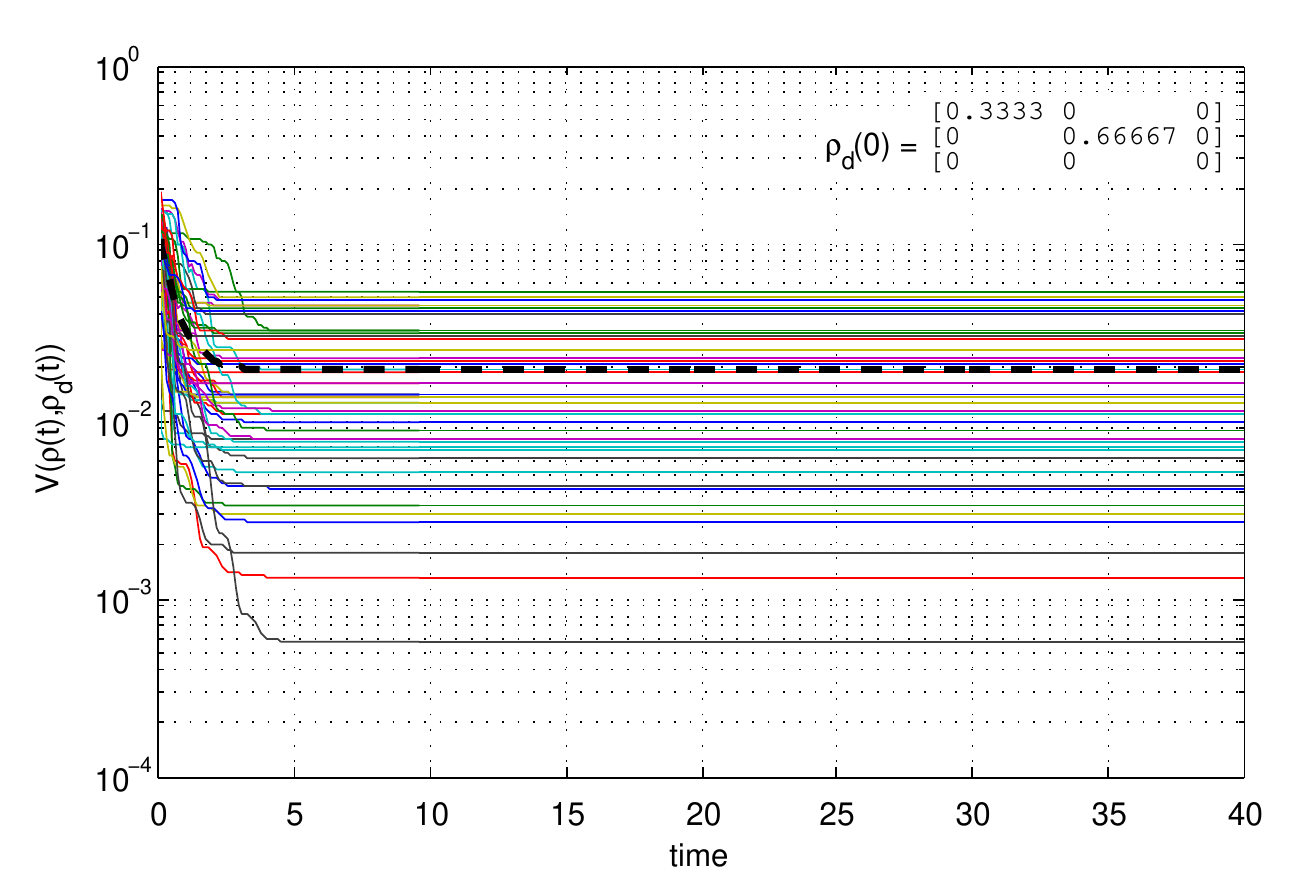}}
\subfloat[Non-stationary target state with irregular $E$]%
  {\includegraphics[width=0.49\textwidth]{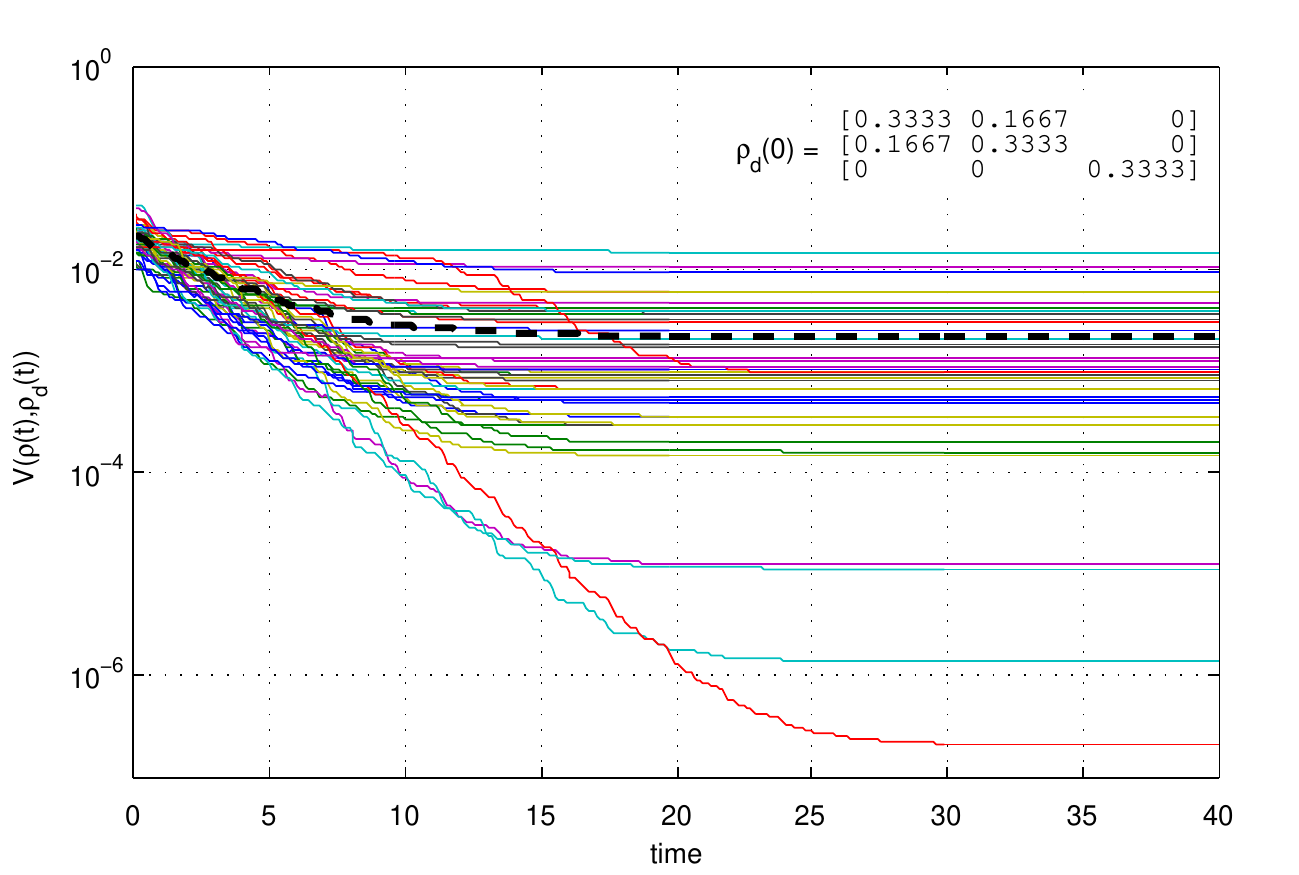}}
\end{figure*}

\subsection{Generic non-stationary target state}

In this case characterizing the invariant set is more complicated as
$E$ may contain points with nonzero diagonal commutators.
\begin{example}
\label{example:1} Let $(\rho(0),\rho_d(0))=(\rho_1,\rho_2)$ with
\begin{eqnarray*}
\rho_1= \begin{bmatrix}
\frac{1}{12} & -\frac{1}{12} & -\frac{1}{12}  \\
-\frac{1}{12} & \frac{11}{24} & \frac{1}{8} \\
-\frac{1}{12} & \frac{1}{8}   & \frac{11}{24}
\end{bmatrix}, \,
\rho_2= \begin{bmatrix}
\frac{1}{3} & -\frac{i}{12} & \frac{i}{12} \\
\frac{i}{12} & \frac{1}{3}  & -\frac{i}{4} \\
-\frac{i}{12} & \frac{i}{4} & \frac{1}{3}
\end{bmatrix}.
\end{eqnarray*}
$\rho_1$ and $\rho_2$ are isospectral and the commutator
$[\rho_1,\rho_2]=\frac{11i}{144}\diag(0,1,-1)$ is diagonal, and thus
$(\rho_1,\rho_2)\in E$, but $[\rho_1,\rho_2]\ne 0$.
\end{example}

When $\rho_d(0)$ is chosen such that $E$ contains points with nonzero
diagonal commutators, Fig.~\ref{fig1}(d) shows that all trajectories
generated by the simulations fail to converge to $\rho_d(t)$, and the
original control design becomes ineffective, even for systems with ideal
Hamiltonians.  Fortunately, however, the above example is quite
exceptional.  We shall see that $E=\{[\rho_1,\rho_2]=0\}$ still holds
for a very large class of generic target states $\rho_d(t)$, and in
these cases Lyapunov control tends to be effective. For convenience, $E$
is called regular if it only contains points with zero commutators, and
irregular otherwise.

Noting that we can write $[\rho_1,\rho_2]=-\Ad_{\rho_2}(\rho_1)$, where
$\Ad_{\rho_2}$ is a linear map from the Hermitian or anti-Hermitian
matrices into $\su(n)$, let $A(\vec{s}_2)$ be the real $(n^2-1)\times
(n^2-1)$ matrix corresponding to the Bloch representation of
$\Ad_{\rho_2}$.  Recall $\su(n)=\T\oplus\C$ and $\RR^{n^2-1}=S_\T\oplus
S_\C$, where $S_\C$ and $S_\T$ are the real subspaces corresponding to
the Cartan and non-Cartan subspaces, $\C$ and $\T$, respectively.  Let
$\tilde{A}(\vec{s}_2)$ be the first $n^2-n$ rows of $A(\vec{s}_2)$
(whose image is $S_\T$).  We have the following lemma, with proof in
Appendix~\ref{app:proof:nonstat:invar}:

\begin{lemma}
\label{lemma:non-sta}
Given a generic $\rho_d(t)$, the LaSalle invariant set $E$ is
irregular if and only if $\rank\tilde{A}(\vec{s}_d(0))<n^2-n$.
\end{lemma}

This lemma provides a necessary and sufficient condition on $\rho_d(0)$
to ensure that $[\rho_1,\rho_2]$ diagonal implies
$[\rho_1,\rho_2]=0$. Assuming the first $n^2-n$ rows correspond to
$S_\T$, let $\tilde{A}_1$ be the submatrix generated from the first
$n^2-n$ rows and last $n^2-n$ columns of $\tilde{A}(\vec{s}_d(0))$.  If
$\det(\tilde{A}_1)\ne 0$ then $\rank \tilde{A}(\vec{s}_d(0))=n^2-n$,
hence $E=\{[\rho_1,\rho_2]=0\}$. We can easily verify that if the
diagonal elements of $\rho_d(0)$ are not equal then $\det(\tilde{A}_1)$$\det[\tilde{A}(\rho_d(0))]=0$.
is a non-trivial polynomial, i.e., $\det(\tilde{A}_1)$ can only have a
finite set of zeros.  Hence:

\begin{theorem}
\label{thm:nonstat:invar} The LaSalle invariant set for a generic
$\rho_d(t)$ is irregular only if $\det[\tilde{A}(\rho_d(0))]=0$ or some
of the diagonal elements of $\rho_d(0)$ are equal.
\end{theorem}

Thus for most generic non-stationary $\rho_d(0)$ we still have $E$
regular.  In this case, given $\rho_d(0)$, let $\rho_d^{(k)}(0)$ be
the $n!$ critical points of $V(\rho)=V(\rho,\rho_d(0))$ with
critical values $V_k$.  Then we can easily see that the $n!$ flows
$(\rho_d^{(k)}(t),\rho_d(t))$ starting from
$(\rho_d^{(k)}(0),\rho_d(0))$ with $f\equiv 0$ are solutions of the
dynamical system satisfying $[\rho_d^{(k)}(t),\rho_d(t)]=0$ for any
$t$, and thus the corresponding trajectories are in the LaSalle
invariant set and are the critical points of $V$ with
$V(\rho_d^{(k)}(t),\rho_d(t))=V_k$, and we can show that any
$\rho(t)$ must converge to one of these critical trajectories. The
trajectories with $\rho_d^{(k)}(t)\neq \rho_d(t)$ cannot be
asymptotically stable as they correspond to unstable critical points
of $V$.  Furthermore, let $V_k$ be the critical values of $V$
ordered in an increasing sequence with $V_0=0$, corresponding to the
global minimum.  Then all initial states $\rho(0)$ with
$V(\rho(0),\rho_d(0))<V_1$ must converge to $\rho_d(t)$ as $V$ is
monotonically decreasing, and thus $\rho_d(t)$ is locally
asymptotically stable.  We can summarize these findings in the
following:

\begin{theorem}
\label{thm:generic:conv1} Given a generic non-stationary $\rho_d(t)$ if
the LaSalle invariant set $E$ is regular, then any trajectory $\rho(t)$
converges to one of the $n!$ critical trajectories $\rho_d^{(k)}(t)$,
$k=1,\ldots,n!$.  All critical trajectories are unstable, except
$\rho_d(t)$, which is locally asymptotically stable, and the global
maximum $\rho^{(n!)}(t)$, which is repulsive.
\end{theorem}

As illustrated in Fig.~\ref{fig1}(b), for regular $E$, all trajectories
$\rho(t)$ in computer simulations keep converging to $\rho_d(t)$.
Although the convergence speed is slow compared to the case of generic
stationary $\rho_d$, the case of regular $E$ is qualitatively different
from the irregular one in Fig~\ref{fig1}(d), where the rate of
convergence drops to zero after some time, resulting in flat-lining of
the trajectories in the semi-logarithmic plot.  We conclude from these
simulations that for a generic non-stationary $\rho_d(t)$ Lyapunov
control is still effective even $E$ is regular, although the convergence
speed may be slow, while when $E$ is irregular, the original control
design becomes ineffective even for systems with ideal Hamiltonians.

\section{(Non)Effectiveness of Lyapunov Control for Non-ideal Systems}
\label{sec:conv_real}

In the previous section we showed that $\rho(t)=\rho_d(t)$ is the only
locally asymptotically stable trajectory if the system is ideal and the
target state is regular.  Realistic systems, unfortunately, often do not
satisfy the strong Hamiltonian requirements, and we now show that in
this case the target state, even if it is stationary, ceases to be a
hyperbolic critical point.  A center manifold emerges and most solutions
do not converge to the target state, rendering the method ineffective.
This clearly illustrates that the dynamics (\ref{eqn:auto}) is very
different from the gradient flow of the Lyapunov function and shows that
the LaSalle invariance principle and critical point analysis of the
Lyapunov function do not suffice to analyze the stability for realistic
systems, and eigenvalue analysis of the linearized system is necessary.
To fully understand the dynamics in this situation we need to analyze it
case by case.  In the following we present an analysis for three-level
systems, which illustrates the techniques that can be applied to other
cases.

\subsection{$H_1$ not fully connected}

Assume $H_0$ still strongly regular but $H_1$ be not fully
connected, for example, consider
\begin{equation*}
H_0 =\begin{bmatrix}
a_1 & 0 & 0  \\
0 & a_2 & 0  \\
0 & 0 & a_3
\end{bmatrix}, \qquad
H_1=\begin{bmatrix}
0 & b_1 & 0  \\
b_1^* & 0 & b_2  \\
0 & b_2^* & 0
\end{bmatrix}
\end{equation*}
where we assume $a_1<a_2<a_3$ and $b_1,b_2\ne 0$.

According to the characterization of the LaSalle invariant set $E$
in Section~\ref{sec:LaSalle}, a necessary condition for
$(\rho_1,\rho_2)$ to be in the invariant set $E$ is
$[\rho_1,\rho_2]$ to be orthogonal to the subspace spanned by
$\B=\Span\{B_m\}_{m=0}^\infty$ with $B_m=\Ad_{-iH_0}^{(m)}(-iH_1)$.
Comparison with (\ref{eq:Bm}) shows that if the coefficient
$b_{k\ell}=0$ then none of the generators $B_m$ have support in the
root space $\T_{k\ell}$ of the Lie algebra, and it is easy to see
that the subspace of $\su(n)$ generated by $\B$ is the direct sum of
all root spaces $\T_{k\ell}$ with $b_{k\ell}\neq 0$. Thus, in our
example, a necessary condition for $(\rho_1,\rho_2)$ to be in the
invariant set $E$ is $[\rho_1,\rho_2] \in \T_{13}\oplus\C$, which
shows that we must have
\begin{equation}
 \label{eq:rho13}
 [\rho_1,\rho_2]= \begin{bmatrix}
  \alpha_{11} & 0 & \alpha_{13}  \\
  0 & \alpha_{22} & 0  \\
  \alpha_{13}^* & 0 & \alpha_{33}
\end{bmatrix}.
\end{equation}
Furthermore, if $(\rho_1,\rho_2)$ is of type~(\ref{eq:rho13}) then
\begin{equation*}
 U_0(t) [\rho_1,\rho_2] U_0(t)^\dagger
 = \begin{bmatrix}
  \alpha_{11} & 0 & e^{i\omega_{13}t}\alpha_{13}  \\
  0 & \alpha_{22} & 0  \\
 e^{-i\omega_{13}t}\alpha_{13}^* & 0 & \alpha_{33}
\end{bmatrix}
\end{equation*}
with $U_0=e^{-iH_0t}$ and $\omega_{k\ell}=a_\ell-a_k$, also has this
form.  Therefore, $[\rho_1,\rho_2]\in \C\oplus\T_{13}$ is a
necessary and sufficient condition for the invariant set $E$.  For
stationary generic $\rho_d$, $E$ consists of all $(\rho_1,\rho_2)$
with $\rho_2=\rho_d$ and
\begin{equation*}
\rho_1= \begin{bmatrix}
       \beta_{11} & 0 & \beta_{13}  \\
       0 & \beta_{22} & 0  \\
       \beta_{13}^* & 0 & \beta_{33}
\end{bmatrix}.
\end{equation*}
Thus, the invariant set $E$ contains $3!=6$ stationary states
corresponding to $\beta_{13}=0$, which coincide with the critical
points of $V(\rho)$, and an infinite number of trajectories with
$\beta_{13}\ne 0$.

We check the stability of linearized system near these fixed points,
concentrating on the local behavior near $\vec{s}_d$.  Working with a
real representation of the linearized system~(\ref{eqn:linear}) and
using the same notation as before, we can still show that
$D_f(\vec{s}_d)$ has $n^2-n$ nonzero eigenvalues, with $n=3$ in our
case.  Since $-iH_1$ has no support in the root space $\T_{13}$, the
$\lambda_{13}$ and $\bar\lambda_{13}$ components of $A_1\vec{s}_d$,
(which correspond to $[-iH_1,\rho_d]$) vanish, and $D_f(\vec{s}_d)$ has
a pair of purely imaginary eigenvalues, whose eigenspaces span the root
space $\T_{13}$, and four eigenvalues with non-zero real parts, which
must be negative as $\vec{s}_d$ is locally stable from the Lyapunov
construction.  However, the existence of two purely imaginary
eigenvalues means that the target state is no longer a hyperbolic fixed
point but there is a center manifold of dimension two. Center manifold
theory shows that the qualitative behavior near the fixed point is
determined by the qualitative behavior of the flows on the center
manifold~\cite{Carr}.  Therefore, the next step is to determine the
center manifold.  For dimensions $>2$ this is usually a non-trivial
problem.  However, if we can find an invariant manifold those tangent
space at $\vec{s}_d$ equals the tangent space of the center manifold,
then this manifold is the center manifold.  In our case solutions in 
the invariant set form a manifold diffeomorphic to the Bloch sphere 
for a qubit system, with the natural embedding
\begin{equation*}
\rho=\begin{bmatrix}
\beta_{11} & 0 & \beta_{13}  \\
0 & \beta_{22} & 0  \\
\beta_{13}^* & 0 & \beta_{33}
\end{bmatrix}
\to \rho'= \frac{1}{\beta_{11}+\beta_{33}}
\begin{bmatrix}
\beta_{11} & \beta_{13}\\
\beta_{13}^* &  \beta_{33}
\end{bmatrix}
\end{equation*}
which maps the state $\rho_d$ (or $\vec{s}_d$) of the qutrit to the
point ${\vec{s}_d}'$ with $\rho_d'=\diag(w_1,w_3)/(w_1+w_3)$, and
the two tangent vectors of the center manifold at $\rho_d$ to the
two tangent vectors of the Bloch sphere at ${\vec{s}_d}'$.  Thus
this manifold is the required center manifold at $\rho_d$ (or
$\vec{s}_d$).  On the center manifold $\rho_d$ is a center with the
nearby solutions cycling around it.  The Hartman-Grobman theorem in
center manifold theory proved by Carr~\cite{Carr} shows that all
solutions outside $E$ converge exponentially to solutions on the
center manifold belonging to $\vec{s}_d$, while the solutions
actually converging to $\vec{s}_d$ only constitute a set of measure
zero.  Therefore, almost all solutions near $\rho_d$ converge to
solutions on the center manifold other than $\rho_d$ and $\rho_d$
becomes no longer asymptotically stable (see Fig.~\ref{fig1}(c)).

\subsection{$H_0$ not strongly regular}

Let $H_1$ fully connected but $H_0$ not strongly regular, e.g.,
\begin{equation*}
H_0= \begin{bmatrix}
0 & 0 & 0  \\
0 & \omega & 0  \\
0 & 0 & 2\omega
\end{bmatrix}, \quad
H_1=\begin{bmatrix}
0 & 1 & 1  \\
1 & 0 & 1  \\
1 & 1 & 0
\end{bmatrix}.
\end{equation*}
Analogously to the section above, we can show that for a given
stationary generic $\rho_d$, the LaSalle invariant set forms a center
manifold with the target state as a center.  Hence, almost all
trajectories near $\rho_d$ converge to other solutions on the center
manifold and $\vec{s}_d$ is not asymptotically stable.

\section{Concluding discussion}
\label{sec:concluding}

We have studied a control design for tracking natural trajectories of
generic quantum states based on the Hilbert-Schmidt distance as a
Lyapunov function.  The analysis shows that the method is effective for
generic density operators if and only if the invariant set is regular,
i.e., contains only trajectories corresponding to critical points of the
Lyapunov function.  Since the Lyapunov function has exactly $n!$ isolated
critical points for generic states, regularity of the invariant set in
this case immediately implies that the target state or trajectory is
isolated and thus locally asymptotically stable, but a detailed analysis
shows that we have almost global convergence in this case.  Although the
set of states that do not converge to the target state is larger than
previously asserted \cite{altafini2}, for ideal systems it is still only
a small subset of the state space, and for stationary target states we
can show that it has measure zero.  When the LaSalle invariant set is
not regular in the other hand, the method not only becomes ineffective,
and the target state ceases to be locally asymptotically stable, but a
center manifold emerges around the target state, which exponentially
attracts all trajectories, preventing convergence to the target state.

The results follow from several steps.  Computation of both the LaSalle
invariant set and the set of critical points of the Lyapunov function
shows that a necessary condition for regularity of the invariant set is
that the system Hamiltonian satisfy certain rather strict conditions,
effectively equivalent to controllability of the linearization.  Further
analysis shows that when we restrict our attention to generic states,
the Lyapunov function is a Morse function with $n!$ isolated critical
points and the target state as the unique global minimum in addition to
a unique global maximum and $n!-2$ saddle points.  The critical points
of the Lyapunov function further correspond to fixed points or critical
trajectories of the dynamical system.  If the dynamical system were a
gradient flow of the Lyapunov function this would allow us to almost
immediately infer almost global convergence to the target states.  As
this is not the case we must analyze the linearization of the dynamics
about the critical points and show that they are hyperbolic.  We do this
rigorously for stationary target states, where the analysis shows that
the $n!$ critical points of the Lyapunov function are indeed hyperbolic
fixed points of the dynamical system if the system Hamiltonian ideal.
For stationary target states this condition also implies for regularity
of the invariant set, and as the Lyapunov function is a Morse function
in our case, we can use it to compute the dimensions of the stable and
unstable manifolds at each of the hyperbolic critical points of the
dynamical system.  This shows that all critical points except the target
state have stable manifolds of dimensions less than the state space and
allows us to conclude that almost all initial states will converge to
the target state in this case.  The flipside of this analysis is that
the target state ceases to be a hyperbolic fixed point of the dynamical
system if the system Hamiltonian is no longer ideal, and in this case a
center manifold emerges around the target state, which exponentially
attracts all trajectories.  For non-stationary target states the method
can fail even if the system Hamiltonian is ideal, for target states that
give rise to a non-regular invariant set, but we also show that such
target states are a measure-zero subset of the state space.

\section*{Acknowledgments}

XW is supported by the Cambridge Overseas Trust and an Elizabeth Cherry
Major Scholarship from Hughes Hall, Cambridge.  SGS is acknowledges
support from an \mbox{EPSRC} Advanced Research Fellowship, Hitachi, a
former Marie Curie Fellowship under EU Knowledge Transfer Programme
MTDK-CT-2004-509223, and NSF Grant PHY05-51164.  We thank Peter
Pemberton-Ross, Tsung-Lung Tsai, Christopher Taylor, Jack Waldron, Jony
Evans, Dan Jane, Yaxiang Yuan, Jonathan Dawes, Lluis Masanes, Rob
Spekkens, Ivan Smith for fruitful discussions, and the editors, Edmond
Jonckheere and Antonio Loria, for constructive suggestions.

\appendix

\subsection{Lie algebra generators}
\label{app:lie_algebra_basis}

\noindent A standard basis for the Lie algebra $\su(n)$ is given by
$\{\lambda_{k\ell},\bar{\lambda}_{k\ell},\lambda_k\}$ for $1\le
k<\ell\le n$, where
\begin{subequations}
\label{eq:lambda}
\begin{align}
 \lambda_k            &=  i(\hat{e}_{kk} - \hat{e}_{k+1,k+1}) \\
 \lambda_{k\ell}      &=  i(\hat{e}_{k\ell}+\hat{e}_{\ell k})\\
 \bar{\lambda}_{k\ell}&=   (\hat{e}_{k\ell}-\hat{e}_{\ell k})
\end{align}
\end{subequations}
and the $(k,\ell)^{\rm th}$ entry of the matrix $\hat{e}_{mn}$ is
$\delta_{km}\delta_{\ell n}$, and $i=\sqrt{-1}$. We have the useful
identities
\begin{subequations}
 \label{eq:lambda_prod}
\begin{align}
  &\Tr(\lambda_{k\ell}\lambda_{k'\ell'})
  = \Tr(\bar\lambda_{k\ell}\bar\lambda_{k'\ell'})
  = -2\delta_{kk'}\delta_{\ell\ell'} \\
  &\Tr(\lambda_{k\ell}\bar\lambda_{k'\ell'})=0
\end{align}
\end{subequations}
and for any diagonal matrix $D=\sum_{k=1}^n d_k\hat{e}_{kk}$
\begin{subequations}
\label{eq:lambda_comm}
\begin{align}
  [D,\lambda_k]             &= 0,\\
  [D,\lambda_{k\ell}]       &= +i(d_k-d_\ell) \bar{\lambda}_{k\ell}, \\
  [D,\bar{\lambda}_{k\ell}] &= -i(d_k-d_\ell) \lambda_{k\ell},
\end{align}
\end{subequations}
The basis~(\ref{eq:lambda}) is not orthonormal but we can define an
equivalent orthonormal basis $\{\sigma_m\}_{m=1}^{n^2-1}$ for
$\su(n)$ by normalizing the $n^2-n$ non-Cartan generators
$\frac{1}{\sqrt{2}}\lambda_{k\ell}$ and
$\frac{1}{\sqrt{2}}\bar{\lambda}_{k\ell}$, and defining the $n-1$
orthonormal generators for the Cartan subalgebra \(
 \sigma_{n^2-n+r} = i [r(r+1)]^{-1/2}
                     \left(\sum_{s=1}^{r} \hat{e}_{ss}
                     -r\hat{e}_{r+1,r+1} \right)
\) for $r=1,\ldots,n-1$.

\subsection{Lemmas in the proof of Theorem \ref{thm:generic:hyperbolic}}
\label{app:proof:generic:hyperbolic}

\begin{lemma}
$D_f(\vec{s}_0)$ vanishes on the subspace $S_\C$.
\end{lemma}

\begin{IEEEproof}
To show that $D_f(\vec{s}_0)\vec{s}=0$ for all $\vec{s}\in S_\C$, it
suffices to show that $A_0 \vec{s}=0$ and $\vec{s}_d^T A_1
\vec{s}=0$ for $\vec{s} \in S_\C$.  $\vec{s}\in S_\C$ corresponds to
density operators $\rho\in i\C$, i.e., $\rho$ diagonal.  As
$A_0\vec{s}$ is the Bloch vector associated with $[-iH_0,\rho]$,
$-iH_0$ is diagonal and diagonal matrices commute, $[-iH_0,\rho]=0$
and $A_0\vec{s}=0$ follows immediately.  To establish the second
part, note that for $i\rho\in\C$ and $-iH_1\in \T$, we have
$[-iH_1,i\rho] \in \T$, or $[-iH_1,\rho]\in i\T$, and $A_1\vec{s}
\in S_\T$.  Since $\rho_d$ is diagonal, i.e., $\vec{s}_d \in S_\C
\perp S_\T$, we have $\vec{s}_d^TA_1\vec{s}=0$ for $\vec{s}\in
S_\C$.
\end{IEEEproof}

\begin{lemma}
The restriction of $D_f(\vec{s}_0)$ to $S_\T$ is well-defined and
its matrix representation $B$ has $n^2-n$ non-zero eigenvalues.
\end{lemma}

\begin{IEEEproof}
Since we already know that $S_\C$ is in the kernel of
$D_f(\vec{s}_0)$, it suffices to show that the image of
$D_f(\vec{s}_0)$ is contained in $S_\T$, i.e.,
$D_f(\vec{s}_0)\vec{s}\in S_\T$.
\begin{align*}
  D_f(\vec{s}_0) \vec{s}
  &= A_0 \vec{s} + A_1 \vec{s}_0  \; \vec{s_d}^T A_1\vec{s} \\
  &= A_0 \vec{s} + (\vec{s_d}^T A_1 \vec{s} ) \, A_1 \vec{s}_0
\end{align*}
shows that it suffices to show that $A_0\vec{s} \in S_\T$ and
$A_1\vec{s}_0\in S_\T$.  Both relations follow from the fact that
the commutator of a Cartan element and a non-Cartan element of the
Lie algebra $\su(n)$ is always in the non-Cartan algebra $\T$, and
thus $[-iH_0,\rho] \in i\T$ since $-iH_0\in \C$, and
$[-iH_1,\rho_d]\in i\T$ since $i\rho_d \in \C$.  Therefore, the
restriction $B: S_\T \to S_\T$ of $D_f(\vec{s}_0) \vec{s}$ is well
defined.

Furthermore, the restriction of $A_0$ to $S_\T$ is a block-diagonal
matrix $B_0=\diag(A_0^{(k,\ell)})$ with
\begin{equation*}
  A_0^{(k,\ell)}= \omega_{k\ell}
  \begin{bmatrix} 0 & 1 \\ -1 & 0 \end{bmatrix}.
\end{equation*}
The restriction $\vec{u}$ of $A_1\vec{s}_0$ to $S_\T$ is a column
vector $(\vec{u}^{(1,2)};\vec{u}^{(1,3)};\ldots;\vec{u}^{(n-1,n)})$
of length $n(n-1)$ consisting of $n(n-1)/2$ elementary parts
\begin{equation}
 \label{eq:ukl}
 \vec{u}^{(k,\ell)} =
\frac{\Delta_{\tau(k)\tau(\ell)}}{\sqrt{2}}
 \begin{bmatrix}
 \Im(b_{k\ell}) \\ \Re(b_{k\ell})
\end{bmatrix}
\end{equation}
for $k=1,\ldots,n-1$ and $\ell=k+1,\ldots,n$.  Similarly, let
$\vec{v}$ be the restriction of $A_1\vec{s}_d$ to $S_\T$.  Then
$\vec{v}=(\vec{v}^{(1,2)};\ldots;\vec{v}^{(n-1,n)})$ with
$\vec{v}^{(k,\ell)}$ as in Eq.~(\ref{eq:ukl}) and $\tau$ the
identity permutation.

Thus the restriction of $D_f(\vec{s}_0)$ to the subspace $S_\T$ is
$B=B_0-\vec{u}\vec{v}^T$.  Since $\omega_{k\ell}\neq 0$ for all
$k,\ell$ by regularity of $H_0$, we have $\det(B_0) = \prod_{k,\ell}
\omega_{k\ell}^2 \neq 0$, i.e., $B_0$ invertible, and \( \det(B) =
\det(B_0 - \vec{u}\vec{v}^T) = (1-\vec{v}^T B_0^{-1} \vec{u})
\det(B_0) \) by the matrix determinant lemma~\cite{matrix}.
$B_0^{-1}$ is block-diagonal with blocks
\begin{equation*}
 C^{(k,\ell)} = [A_0^{(k,\ell)}]^{-1}
              = \frac{1}{\omega_{k\ell}}
 \begin{bmatrix} 0 & -1 \\ 1 & 0 \end{bmatrix}.
\end{equation*}
Thus $\vec{v}^T B_0^{-1}\vec{u}=\sum_{k,\ell} [\vec{v}^{(k,\ell)}]^T
C^{(k,\ell)} \vec{u}^{(k,\ell)}$ vanishes as
\begin{equation*}
  (\Im(b_{k\ell}),\Re(b_{k\ell})
 \begin{bmatrix} 0 & -1 \\ 1 & 0 \end{bmatrix}
 \begin{bmatrix}
 \Im(b_{k\ell}) \\ \Re(b_{k\ell})
\end{bmatrix} = 0, \quad \forall k, \ell.
\end{equation*}
Hence, $\det(B)=\det(B_0)\neq0$, and the restriction of
$D_f(\vec{s}_0)$ to $S_\T$ has only non-zero eigenvalues.
\end{IEEEproof}

\begin{lemma}
\label{lemma:generic:3} If $i\beta$ is a purely imaginary eigenvalue
of $B$ then it must be an eigenvalue of $B_0$, i.e., $i\beta=\pm
i\omega_{k\ell}$ for some $(k,\ell)$, and either the associated
eigenvector $\vec{e}$ must be an eigenvector of $B_0$ with the same
eigenvalue, or the restriction of $A_1\vec{s}_0$ to the $(k,\ell)$
subspace must vanish.
\end{lemma}

\begin{IEEEproof}
If $i\gamma$ is not an eigenvalue of $B_0$ then $(B_0-i\beta I)$ is
invertible and by the matrix determinant lemma
\begin{align*}
  0 &= \det(B_0 - \vec{u}\vec{v}^T - i\beta I) \\
    &= \det( (B_0-i\beta I) - \vec{u}\vec{v}^T) \\
    &= (1-\vec{v}^T (B_0-i\beta I)^{-1} \vec{u}) \det(B_0-i\beta I).
\end{align*}
Since $\det(B_0-i\beta I)\neq 0$ we must therefore have
\begin{equation*}
  \vec{v}^T (B_0-i\beta I)^{-1} \vec{u} = 1.
\end{equation*}
Noting $(B_0-i\beta I)^{-1}$ is block-diagonal with blocks
\begin{equation}
 \label{eq:B0inv}
  C_\beta^{(k,\ell)} =
  \frac{1}{\omega_{k\ell}^2-\beta^2}
  \begin{bmatrix}
   -i\beta & -\omega_{k\ell} \\ \omega_{k\ell} & -i\beta
  \end{bmatrix},
\end{equation}
\begin{equation*}
  \Big(\Im(b_{k\ell}),\Re(b_{k\ell}\Big)
 \begin{bmatrix} -i\beta & -\omega_{k\ell} \\ \omega_{k\ell} & -i\beta
\end{bmatrix}
 \begin{bmatrix}
 \Im(b_{k\ell}) \\ \Re(b_{k\ell})
\end{bmatrix} = -i\beta |b_{k\ell}|^2
\end{equation*}
for all $k,\ell$, this leads to
\begin{align*}
 1& =\vec{v}^T (B_0-i\beta I)^{-1}\vec{u}
    =\sum_{k,\ell}[\vec{v}^{(k,\ell)}]^T C_\beta^{(k,\ell)}\vec{u}^{(k,\ell)}
\\
  & = \frac{-i\beta}{2} \sum_{k,\ell}
    \frac{\Delta_{k\ell}\Delta_{\tau(k)\tau(\ell)}}{\omega_{k\ell}^2-\beta^2}
    |b_{k\ell}|^2.
\end{align*}
Since all terms in the sum are real this is a contradiction.  Thus
if $i\beta$ is a purely imaginary eigenvalue of $B$ then it must be
an eigenvalue of $B_0$.

Since the spectrum of $B_0$ is $\{\pm i\omega_{k\ell}\}$, this means
$i\beta =\pm i\omega_{k\ell}$ for some $(k,\ell)$.  Without loss of
generality assume $\gamma=\omega_{12}>0$ and let
$\vec{e}=\vec{x}+i\vec{y}$ be the associated eigenvector of $B$.
Then
\begin{align}
\label{eqn:Bx0}
  B\vec{e} = (B_0-\vec{u}\vec{v}^T)(\vec{x}+i\vec{y})
           = i\omega_{12}(\vec{x}+i\vec{y}),
\end{align}
which is equivalent to
\begin{subequations}
\label{eq:Bx}
\begin{align}
 (B_0-\vec{u}\vec{v}^T)\vec{x}&= -\omega_{12} \vec{y}\\
 (B_0-\vec{u}\vec{v}^T)\vec{y}&=  \omega_{12} \vec{x}.
\end{align}
\end{subequations}
Multiplying~(\ref{eq:Bx}b) by $-\omega_{12} B_0^{-1}$ and adding it
to (\ref{eq:Bx}a)
\begin{align*}
 \underline{B_0\vec{x}}-\vec{u}\vec{v}^T\vec{x}
 +\omega_{12}B_0^{-1}\vec{u}\vec{v}^T\vec{y}
 &= \underline{-\omega_{12}^2 B_0^{-1} \vec{x}}
\end{align*}
Eq.~(\ref{eq:B0inv}) shows that
$-\omega_{12}^2[B_0^{(1,2)}]^{-1}=B_0^{(1,2)}$, i.e., on the
$\T_{12}$ subspace the underlined terms above cancel, and thus the
first two rows of the above system of equations are
\begin{align*}
 \begin{bmatrix} u_1 \\ u_2 \end{bmatrix}
 (\vec{v}^T\vec{x})
= \begin{bmatrix} 0 & -1\\ 1 & 0 \end{bmatrix}
 \begin{bmatrix} u_1 \\ u_2 \end{bmatrix}
 (\vec{v}^T\vec{y}).
\end{align*}
If $\vec{v}^T\vec{x}\neq 0$ then the last equation gives
$u_1=-c^2u_1$ and $u_2=-c^2u_2$ with
$c=\vec{v}^T\vec{y}/\vec{v}^T\vec{x}$, which can only be satisfied
if $u_1=u_2=0$.  Similarly if $\vec{v}^T\vec{y}\neq0$. If
$\vec{v}^T\vec{x}=\vec{v}^T\vec{y}=0$ then we have
$B\vec{e}=B_0\vec{e}=i\omega_{12}\vec{e}$, implying that $\vec{e}$
is an eigenvector of $B_0$ associated with $i\omega_{12}$.
\end{IEEEproof}

\subsection{Proof of Lemma~\ref{lemma:non-sta}}
\label{app:proof:nonstat:invar}

\begin{lemma}
Given a generic $\rho_d(t)$, the LaSalle invariant set $E$ is
irregular if and only if $\rank\tilde{A}(\vec{s}_d(0))<n^2-n$.
\end{lemma}

\begin{IEEEproof}
It suffices to show that if $\rank \tilde{A}(\vec{s}_d)=n^2-n$, then
for any $\rho$ such that $[\rho,\rho_d(0)]$ diagonal, we have
$[\rho,\rho_d(0)]=0$.  If this is true then for any
$(\rho_1,\rho_2)\in E$ with diagonal commutator, there exists some
$t_0$ such that $\rho_2=e^{iH_0t_0}\rho_d(0)e^{-iH_0t_0}$ and since
$[\rho_1,\rho_2]$ is diagonal, $[e^{-iH_0t_0}\rho_1
e^{iH_0t_0},\rho_d(0)]$ is also diagonal, hence equal to zero and
$[\rho_1,\rho_2]=0$.

Let $\rho_2=\rho_d(0)$.  First we show that the kernel of
$A(\vec{s}_2)$ has dimension $n-1$ and thus $\rank A(\vec{s}_2) \le
n^2-n$.  In this case $\rank \tilde{A}(\vec{s}_d) = n^2-n = \rank
A(\vec{s}_2)$ implies that the remaining $n-1$ rows of
$A(\vec{s}_2)$ are linear combinations of the rows of
$\tilde{A}(\vec{s}_2)$ and therefore
$\tilde{A}(\vec{s}_2)\vec{s}_1=\vec{0}$ implies
$A(\vec{s}_2)\vec{s}_1=\vec{0}$, or $[\rho_1,\rho_2]=0$.

In order to show that the kernel of $A(\vec{s}_2)$ has dimension
$n-1$, we recall that if $\rho_2=U \diag(w_1,\ldots,w_n) U^\dagger$
for some $U\in\SU(n)$ then $[\rho_1,\rho_2]=0$ for all
$\rho_1=U\diag(w_{\tau(1)},\ldots,w_{\tau(n)}) U^\dagger$, where
$\tau$ is a permutation of $\{1,\ldots,n\}$.  If the $w_k\ge 0$ are
distinct then these $\rho_1$'s span at least a subspace of dimension
$n$ since the determinant of the circulant matrix
\begin{equation*}
  C = \begin{bmatrix}
       w_1 & w_2 & \ldots & w_{n-1} & w_n \\
       w_2 & w_3 & \ldots & w_n & w_1 \\
       \vdots & \vdots & \ddots & \vdots & \vdots \\
       w_{n-1} & w_n & \ldots & w_{n-3} & w_{n-2} \\
       w_n & w_1 & \ldots & w_{n-2} & w_{n-1}
      \end{bmatrix}
\end{equation*}
is non-zero, and hence its columns are linearly independent and span
the $n$-dimensional subspace of diagonal matrices.  If the $w_k$ are
distinct then the kernel cannot have dimension greater than $n-1$
since the $\rho_1$ can only span a subspace isomorphic to the set of
diagonal matrices. Thus, the kernel of $A(\vec{s}_2)$ has dimension
$n-1$.  (The dimension is reduced by one since we drop the
projection of $\rho$ onto the identity in the Bloch representation.)
Similarly, we can prove if $\rank \tilde{A}(\vec{s}_d(0))<n^2-n$,
then $E$ contains points with nonzero commutator.
\end{IEEEproof}

\begin{IEEEbiography}{Xiaoting Wang}
is a PhD student in the Dept of Applied Mathematics and Theoretical
Physics at the University of Cambridge.  He received a Certificate of
Advanced Study in Mathematics from the University of Cambridge in 2006,
and a BSc degree in mathematics from Wuhan University, China, in 2005.
His research interests include the theory of quantum control and its
applications in quantum information science.
\end{IEEEbiography}

\begin{IEEEbiography}{S.~G.~Schirmer}
is an Advanced Research Fellow of the UK Engineering and Physical
Sciences Research Council (EPSRC) and has held positions as Marie Curie
Senior Research Fellow, Research Fellow of the Cambridge-MIT Institute,
and Coordinator of the Quantum Technologies Group.  Her research
interests include nano-science at the quantum edge and quantum
engineering, especially modeling, control and characterization of
quantum systems and devices.
\end{IEEEbiography}

\end{document}